\documentclass[aps,preprintnumbers,nofootinbib]{revtex4}
\usepackage{epsfig,graphics}
\usepackage{amsfonts,amssymb,amsmath}            
\usepackage{ifthen}
\usepackage{amsthm}
\usepackage{hyperref}
\usepackage{xcolor}
\definecolor{mylinkcolor}{rgb}{0,0,0.8} 
\hypersetup{unicode=true, %
  bookmarksnumbered=false,bookmarksopen=false,bookmarksopenlevel=1, %
  breaklinks=true,pdfborder={0 0 0},colorlinks=true}%
\hypersetup{%
  anchorcolor=mylinkcolor,citecolor=mylinkcolor, %
  filecolor=mylinkcolor,linkcolor=mylinkcolor, %
  menucolor=mylinkcolor,runcolor=mylinkcolor, %
  urlcolor=mylinkcolor} 

\usepackage[nodayofweek]{datetime}

\newcommand{\comment}[1]{}
\newcommand{\ket}[1]{| #1 \rangle}
\newcommand{\bra}[1]{\langle #1 |}

\newcommand{\ketbra}[2]{|#1\rangle\!\langle#2|}
\newcommand{\proj}[1]{\ketbra{#1}{#1}}

\newcommand{\tr}{{\rm tr}}

\newcommand{\cL}{\mathcal{L}}
\newcommand{\cP}{\mathcal{P}}
\newcommand{\cQ}{\mathcal{Q}}
\newcommand{\cS}{\mathcal{S}}

\newcommand{\cNS}{\mathcal{NS}}
\newcommand{\rL}{\mathrm{L}}

\newcommand{\ot}{\otimes}

\newcommand{\eps}{\epsilon}
\theoremstyle{plain}
\newtheorem{theorem}{Theorem}
\newtheorem{lemma}[theorem]{Lemma}
\newtheorem{corollary}[theorem]{Corollary}

\theoremstyle{definition}
\newtheorem{definition}{Definition}
\newtheorem{remark}{Remark}

\begin{document}

\title{Bell Inequalities From No-Signalling Distributions} 
\date{$27^{\text{th}}$ June 2019}

\author{Thomas \surname{Cope}}
\email[]{thomas.cope@york.ac.uk}
\affiliation{Department of Mathematics, University of York,
  Heslington, York, YO10 5DD, UK}
\author{Roger \surname{Colbeck}}
\email[]{roger.colbeck@york.ac.uk}
\affiliation{Department of Mathematics, University of York,
  Heslington, York, YO10 5DD, UK}

\begin{abstract}
  A Bell inequality is a constraint on a set of correlations whose
  violation can be used to certify non-locality.  They are
  instrumental for device-independent tasks such as key distribution
  or randomness expansion. In this work we consider bipartite Bell
  inequalities where two parties have $m_A$ and $m_B$ possible inputs
  and give $n_A$ and $n_B$ possible outputs, referring to this as the
  $(m_A,m_B,n_A,n_B)$ scenario. By exploiting knowledge of the set of
  extremal no-signalling distributions, we find all 175 Bell
  inequality classes in the $(4,4,2,2)$ scenario as well as providing
  a partial list of 18277 classes in the $(4,5,2,2)$ scenario. We also
  use a probabilistic algorithm to obtain 5 classes of inequality in
  the $(2,3,3,2)$ scenario, which we confirmed to be complete, 25
  classes in the $(3,3,2,3)$ scenario, and a partial list of 21170
  classes in the $(3,3,3,3)$ scenario.  Our inequalities are given in
  supplementary files.  Finally, we discuss the application of these
  inequalities to the detection loophole problem, and provide new
  lower bounds on the detection efficiency threshold for small numbers
  of inputs and outputs.
\end{abstract}

\maketitle

\section{Introduction}
Bell inequalities~\cite{B1964} can be thought of as constraints on the
set of correlations realisable at spacelike separation by using shared
classical randomness and freely chosen measurements.  One of the most
counterintuitive features of quantum theory is that by sharing quantum
systems instead of classical randomness, these inequalities can be
violated, a fact that has been subject to extensive experimental
investigation~\cite{Aspect81,Tittel1998,Giustina&,Hensen&,Shalm&}.
This curious feature has since been used for
cryptography~\cite{Ekert91} and shown to enable device-independent
information processing tasks, such as quantum key
distribution~\cite{MayersYao,BHK,AGM,VV2} and randomness
expansion~\cite{ColbeckThesis,PAMBMMOHLMM,CK2,MS1}.  In essence, that
shared classical randomness cannot explain the observed correlations
implies that any eavesdropper can have only limited information about
them, and so key or random numbers can be distilled.

In spite of their usefulness, relatively little is known about the set
of Bell inequalities in all but the simplest cases.  In part, this is
due to the complexity of finding them and the fact that the number of
such inequalities grows rapidly as the number of inputs or outputs is
increased.  Bell inequalities can be thought of in a geometric way as
the facets of the polytope of local (classically realisable)
correlations. This insight means that Bell inequalities can in
principle be found by facet enumeration, a well-known problem in
polytope theory.  However, because it quickly becomes intractable,
this method has only been applied to enumerate all Bell inequalities in
simple cases.

In this paper we propose some alternative algorithms. These algorithms
are relatively easy to run, although they have the disadvantage that
they don't give a certificate when all the Bell inequalities have been
found.  In several cases we have been able to enumerate the complete
list of Bell inequality classes\footnote{Two Bell inequalities are in
  the same class if they are related by a relabelling.}, whilst in
other cases we find lower bounds on the total number of classes of
Bell inequality, which informs the complexity of doing facet
enumeration in these cases. Our approach utilises knowledge of the set
of extremal no-signalling distributions, as opposed to techniques
based on enumerating the facets of simpler, related
polytopes~\cite{BGP2010, PV2009}, or shelling techniques which travel
along the polytope's edges~\cite{PV2009}.

In particular, we provide a complete list of the 175 classes of Bell
inequalities for the $(4,4,2,2)$ scenario, normalized in such a way as
to provide easy comparison. The number of inequalities for this
scenario was already known~\cite{DS2015} but a useable list was not
provided, unlike for the $(2,2,2,2)$~\cite{CHSH1969,F1981},
$(2,m,2,n)$~\cite{P2004}, $(3,3,2,2)$~\cite{F1981,S2003,CG2004},
$(3,4,2,2)$~\cite{CG2004}, $(3,5,2,2)$~\cite{CG2004,QVB2014} and
$(2,2,3,3)$~\cite{CGLMP2002,KKCZO2002} scenarios, which we summarize
in Table~\ref{solvedcases}.

In addition we investigate a number of other scenarios.  In the
$(2,3,3,2)$ scenario we find 5 classes of Bell inequality, which we
confirmed to be complete using facet enumeration.  In the $(3,3,2,3)$
scenario we find 25 classes of Bell inequality, which may be
complete. In the $(4,5,2,2)$ scenario we find 18277 classes of
inequality. Finally, in the $(3,3,3,3)$ scenario we find 21170
classes, adding substantially to the set of 19 known classes for this
scenario (one was found in~\cite{CG2004} and 18
in~\cite{SBSL2016}). In the last two cases, the sets of classes found
are incomplete.  Explicit representatives of each class are given as
supplementary files.

\begin{table}
\begin{tabular}{c|c|c|c}
Scenario & Number of Inequality Classes & Number of Facets & Reference \\
\hline
$(2,2,2,2)$ & 2 & 24 &\cite{CHSH1969,F1981}\\
$(2,m_B,2,n_B)$ & 2 &$2(2^{n_B}-2)(m_B^2-2)+4n_Bm_B$  &\cite{P2004}\\
$(2,2,3,3)$ & 3 & 1116 &\cite{CGLMP2002,KKCZO2002,CG2004}\\
$(3,3,2,2)$ & 3 & 684 &\cite{F1981,S2003,CG2004} \\
$(2,3,3,2)$ & 5 & 1260 & Sec.~\ref{2332Sec} \\
$(3,4,2,2)$ & 6  & 12480 &\cite{DS2015,CG2004} \\
$(3,5,2,2)$ & 7  & 71340 &\cite{DS2015,QVB2014} \\
$(4,4,2,2)$ & 175 & 36391264 &\cite{DS2015}, Sec.~\ref{4422Sec} \\
\hline
$(3,3,2,3)$ & $\geq 25$ & $\geq 252558$ & Sec.~\ref{3323Sec}\\
$(4,5,2,2)$ & $\geq 18277$ && Sec.~\ref{4522Sec}\\
$(3,3,3,3)$ & $\geq 21170$ && Sec.~\ref{3333Sec}
\end{tabular}
\caption{Known Bell inequality classes for bipartite scenarios. For
  all scenarios, one of the classes always corresponds to the trivial
  positivity inequalities. Those in the upper section are known to be
  complete, whilst for the final three we provide lower bounds on the
  number of classes. We suspect the $(3,3,2,3)$ may be complete (see
  Sec.~\ref{3323Sec} for details).}\label{solvedcases}
\end{table}
The structure of the paper is as follows. In Section~\ref{sec:prelim}
we give an overview of relevant polytope theory and the related topic of
linear programming, before discussing Bell inequalities and detailing
our representation of them. In Section~\ref{sec:gen}, we discuss how
to exploit knowledge of the set of extremal no-signalling
distributions to obtain new Bell inequalities, as well as a technique
for doing so without such knowledge. Section~\ref{sec:results} then
gives the results.  Finally, in Section~\ref{sec:det} we apply these
new inequalities to the problem of the detection loophole, presenting
some numerical results and list the new inequalities which are the
most promising candidates for lowering the detection threshold for
small numbers of inputs and outputs.

\section{Preliminaries}\label{sec:prelim}
\subsection{Polytope Theory}\label{polysec}
A polytope is a convex set that can be described by the intersection
of a finite number of half-spaces\footnote{For a detailed summary of
  polytope theory, see~\cite{G1967}.}.  Given $A\in\mathbb{R}^{r\times
  t}$ and $\mathbf{c}\in\mathbb{R}^r$ we can write
\begin{equation}
\cP=\left\{\mathbf{x}\in\mathbb{R}^t\,\middle\vert\,A\mathbf{x}\geq \mathbf{c}\right\}
\end{equation}
(each $A_k\mathbf{x}\geq c_k$ describes a half-space).  This is called
an \emph{H-representation} of the polytope.  Polytopes may also be
described using a \emph{V-representation}. If a
polytope is bounded, then for some set $\{\mathbf{x}_k\}$ with
$\mathbf{x}_k\in\mathbb{R}^t$ it can be written
\begin{equation}\label{convex}
\cP=\left\{\mathbf{x}=\sum_k\lambda_k\mathbf{x}_k\,\middle\vert\,\sum_k\lambda_k=1,\;\lambda_k\geq 0\right\}.
\end{equation}
According to the Minkowski-Weyl theorem, every polytope
admits both a V-representation and H-representation.  We will always
deal with \emph{minimal representations} (in which unnecessary
half-spaces or points ${\bf x}_i$ are removed).  In a minimal
representation $\{\mathbf{x}_k\}$ are \emph{vertices}\footnote{For an
  unbounded polytope the V-representation will also have
  \emph{rays}.}.  Given a polytope $\cP\subset\mathbb{R}^t$ with
dimension $d\leq t$, the intersection of $\cP$ with a bounding
hyperplane $A_k\mathbf{x}=c_k$ is a \emph{facet} of the polytope if it
has dimension $d-1$. The minimal H-representation is exactly the set of facet-defining hyperplanes. 

Converting from V-representation to H-representation is known as
\textit{facet enumeration} whilst going from H-representation to
V-representation is known as \textit{vertex enumeration}. For a
polytope of dimension $d$, given a V-representation with $n$ vertices
(or an H-representation with $n$ half-spaces) there are algorithms that
can perform this conversion in time $O(ndr)$, with $r$ the number of
facets (vertices) enumerated~\cite{AF1992}\footnote{Some improvements
  on this have been achieved for specific classes of polytope.}. When
performing facet enumeration, $r$ is generally not known in advance,
and hence the worst case scenario is often used to provide an upper
bound. For a given dimension and number of vertices, the so called
\emph{cyclic polytope}~\cite{MS1971} has the largest possible number
of facets,
$\binom{n-\lfloor\frac{d-1}{2}\rfloor}{n-d}+\binom{n-\lfloor\frac{d-2}{2}\rfloor}{n-d}$. Using
this we obtain complexity
$O(n^{\lfloor\frac{d}{2}\rfloor})$~\cite{MS1971}. By contrast the
simplest polytope is the simplex, with only $d+1$ facets.

In this work it is convenient to represent points using matrices
rather than vectors.  In this case, the H-representation of a polytope
is based on a set $\{B_i\}$ with $B_i\in\mathbb{R}^{s\times t}$ and a
set of real numbers $c_i$ and can be expressed as
\begin{equation}
  \left\{\Pi\in\mathbb{R}^{s\times t}\,\middle\vert\,\tr(B_i^T\Pi)\geq
    c_i\,\forall\, i\right\}.
\end{equation}
Likewise the V-representation is formed via a set $\{\Pi_k\}$,
$\Pi_k\in\mathbb{R}^{s\times t}$ as
\begin{equation}
  \left\{\Pi=\sum_{k=1}^n\lambda_k\Pi_k\,\middle\vert\,\sum^n_{k=1} \lambda_k=1,\;\lambda_k\geq 0\right\}.
\end{equation}

\subsection{Linear Programming}
A linear programming problem involves the optimization of a linear
objective function over a set of variables constrained by a finite
number of linear equalities and/or inequalities~\cite{W1971}.  The
canonical form of a linear programming problem is as follows: given a
fixed $\mathbf{c}\in\mathbb{R}^n,\mathbf{q}\in\mathbb{R}^d$ and
$G\in\mathbb{R}^{d\times n}$,
\begin{equation}
\max_{\mathbf{x}}\  \mathbf{c}^T\mathbf{x}\quad\text{subject to}\quad G\mathbf{x} \leq \mathbf{q},\;\mathbf{x}\in\mathbb{R}^n,\;\mathbf{x}\geq \mathbf{0}\,.
\end{equation}
For our later considerations it is convenient to rewrite this using
$A_i,Q\in\mathbb{R}^{s\times t}$ as
\begin{align}\label{primallin}
  \max_{\{x_i\}}\ \sum_ic_ix_i\quad\text{subject to}\quad\sum_ix_iA_i\leq Q,\;x_i \geq 0\, \forall\,i\,.
\end{align}
We refer to this as the \textit{primal} form. Any linear
programming problem can be written in this way~\cite{W1971}.

A linear programming problem is said to be \emph{infeasible}, if there
is no $\{x_i\}$ satisfying the constraints. Otherwise we call the
domain of $\{x_i\}$ satisfying the constraints the \emph{feasible
  region}. If the problem admits a finite solution, the problem is
\emph{bounded}. If the domain is bounded, then it forms a convex
polytope, and the \emph{maximum principle}~\cite{R1970} states that
the optimum value is achieved at an extremal point of the polytope and
is finite.  Every linear program has an associated dual. For a problem
written in the form~\eqref{primallin}, the dual problem is
\begin{align}\label{duallin}
  \min_M\ \tr(M^TQ)\quad\text{subject to}\quad\tr(M^TA_i)\geq c_i \text{ for all }i,\;M\in\mathbb{R}^{s\times t},\;M\geq 0\,,
\end{align}
where the condition $M\geq0$ should be understood elementwise.

Linear programming problems are \emph{strongly dual}: an optimum
solution $\{x_i^*\}$ for the problem~\eqref{primallin}, and
$M^*$ for the problem~\eqref{duallin} satisfy
$\sum_ic_ix^*_i=\tr((M^*)^TQ)$. If the
primal is unbounded, then the dual is infeasible and vice
versa. Solutions $\{x_i^*\},M^*$ also satisfy the
\emph{complementary slackness} conditions~\cite{W1971}, which in our
notation are
\begin{align}
\tr\left((M^*)^T\left(Q-\sum_ix^*_iA_i\right)\right)&=0,\label{Slack1}\\
\sum_i(\tr((M^*)^TA_i)-c_i)x^*_i&=0\,.\label{Slack2}
\end{align}

Two common approaches for solving linear programming problems are
simplex algorithms~\cite{D1947} and interior point
methods~\cite{MRR2006}.  Simplex algorithms exploit the fact that the
optimum of a linear program is always achieved at a vertex. Such
algorithms move between vertices by following edges that improve the
value of the objective function until no further improvement is
possible. Interior point methods make successive steps towards the
optimal solution while remaining in the interior of the feasible
region. Since we are interested in finding extremal Bell inequalities,
simplex algorithms will be most useful for us.

\subsection{Representing Probability Distributions and Bell Inequalities.}\label{probandbell}
In this work we focus on the bipartite case, although most of the
techniques we consider generalise straightforwardly to multipartite
scenarios.  We consider two spacelike separate measurements, modelling
the inputs using random variables $X$ and $Y$, and the respective
outputs $A$ and $B$.  We label the possible values of $X$ by
$\{1,\ldots,m_A\}$.  Likewise, $Y$ takes values from
$\{1,\ldots,m_B\}$, $A$ takes values from $\{1,\ldots,n_A\}$ and $B$
takes values from $\{1,\ldots,n_B\}$.  Sometimes we will consider
cases in which different inputs have different numbers of outputs.
Taking the notation from~\cite{RBG2014}, if the measurements
$X=1,2,\ldots$ have numbers of outcomes $n_A^1,n_A^2,\ldots$, and
likewise measurements $Y=1,2,\ldots$ have numbers of outcomes
$n_B^1,n_B^2,\ldots$ we will label the scenario
$[(n_A^1\ n_A^2\ \ldots)\ (n_B^1\ n_B^2\ \ldots)]$.  If the
measurements $X=x$ and $Y=y$ are performed\footnote{We use upper case
  to denote random variables and lower case for particular values of
  these.}, the joint distribution over the outputs is $P_{AB|xy}$,
i.e., for all $x,y,a,b$ we have $0\leq P_{AB|xy}(a,b)\leq 1$, and for
all $x,y$ we have $\sum_{a,b}P_{AB|xy}(a,b)=1$.  A distribution is said
to be no-signalling if it satisfies
\begin{align*}
\sum_bP_{AB|xy}(a,b)=\sum_b P_{AB|xy'}(a,b)\quad\forall\ a,x,y,y',\quad\text{and}\quad\sum_aP_{AB|xy}(a,b)=\sum_a P_{AB|x'y}(a,b)\quad\forall\ b,y,x,x'.
\end{align*}
Since we consider measurements made at spacelike separation, all
distributions will be no-signalling.

Using notation from Tsirelson~\cite{T1993}, we express the conditional
distribution $P_{AB|XY}$ using an $m_An_A\times m_Bn_B$ matrix:
\begin{equation}\label{tableform}
\Pi=\left(\begin{array}{ccc|cccc|ccc}
P_{AB|11}(1,1) & \ldots & P_{AB|11}(1,n_B) & \ldots&\ldots&\ldots&\ldots & P_{AB|1m_B}(1,1) & \ldots & P_{AB|1m_B}(1,n_B) \\
\vdots & \ddots & \vdots & \ldots&\ldots&\ldots&\ldots & \vdots & \ddots & \vdots \\
P_{AB|11}(n_A,1) & \ldots & P_{AB|11}(n_A,n_B) & \ldots&\ldots&\ldots&\ldots & P_{AB|1m_B}(n_A,1) & \ldots & P_{AB|1m_B}(n_A,n_B)\\
\hline
\vdots & \vdots & \vdots & \ddots &&& & \vdots & \vdots &\vdots\\
\vdots & \vdots & \vdots & &&& & \vdots & \vdots &\vdots\\
\hline
P_{AB|m_A1}(11) & \ldots & P_{AB|m_A1}(1,n_B) & \ldots&\ldots&\ldots&\ldots & P_{AB|m_Am_B}(1,1) & \ldots & P_{AB|m_Am_B}(1,n_B) \\
\vdots & \ddots & \vdots & \ldots&\ldots&\ldots&\ldots & \vdots & \ddots & \vdots \\
P_{AB|m_A1}(n_A,1) & \ldots & P_{AB|m_A1}(n_A,n_B) & \ldots&\ldots&\ldots&\ldots &P_{AB|m_Am_B}(n_A,1)&\ldots &P_{AB|m_Am_B}(n_A,n_B)
\end{array}\right),
\end{equation}
where, for clarity, we have added dividing lines to the matrices to
indicate the different measurements. This notation makes it convenient
to check the no-signalling conditions.  However, it has some
redundancy: the no-signalling conditions and the normalization
condition mean that the true dimension of the space is
$(m_A(n_A-1)+1)(m_B(n_B-1)+1)-1$~\cite{T1993}.

A \emph{local deterministic} distribution is one for which
$P_{AB|XY}=P_{A|X}P_{B|Y}$ and for which $P_{A|x}(a)\in\{0,1\}$ for
all $a,x$ and $P_{B|y}(b)\in\{0,1\}$ for all $b,y$.  There are
$(n_A)^{m_A}(n_B)^{m_B}$ such distributions, and we use
$P_{AB|XY}^{\rL,i}$ to denote the $i^{\text{th}}$ local distribution for
$i=1,\ldots,(n_A)^{m_A}(n_B)^{m_B}$.  A \emph{local} distribution is then one that can be
written as a convex combination of local deterministic distributions,
i.e., $P_{AB|XY}=\sum_i\lambda_iP_{AB|XY}^{\rL,i}$, where
$\lambda_i\geq0$ and $\sum_i\lambda_i=1$.  We use
$\cL_{(m_A,m_B,n_A,n_B)}$ for the set of local distributions in each
scenario.  For all $(m_A,m_B,n_A,n_B)$, these form a polytope (the
\emph{local polytope}) with the local deterministic distributions as
its vertices.

A \emph{Bell inequality} is a linear inequality that is satisfied if
and only if a distribution is local.  The most important class of
these are the \emph{facet} Bell inequalities, which represent the
facets of the local polytope.  In principle, these can be found by
facet enumeration starting from the local deterministic
distributions.  We can represent every Bell inequality using a
$m_An_A\times m_Bn_B$ matrix, $B$, such that the Bell inequality can
be expressed as $\tr(B^T\Pi)\geq c$, where $c$ is some constant.  If
$\Pi$ is local, then it necessarily satisfies every Bell
inequality. Given a complete set of facet Bell inequalities, a
distribution is local if and only if it satisfies them all. Thus,
finding all the facet Bell inequalities is an important task.

Note that the facet Bell inequalities include $m_Am_Bn_An_B$
``trivial'' inequalities that correspond to the positivity of each
component of $\Pi$, and are necessary for $\Pi$ to represent a
probability distribution.  For each non-trivial Bell inequality there
is some no-signalling distribution that violates it.

Note that there are many ways to express the same Bell inequality.  To
explain these, we introduce a little terminology.  We say that a
matrix $M$ is \emph{no-signalling type} if for all $\Pi$ representing
a no-signalling distribution we have $\tr(M^T\Pi)=0$.  We say that a
matrix $M$ is \emph{identity type} if for all $\Pi$ representing a
valid distribution $\tr(M^T\Pi)=1$. (for example, for
$(m_A,m_B,n_A,n_B)=(2,2,2,2)$,
$\left(\begin{array}{cc|cc}1&1&-1&-1\\0&0&0&0\\\hline0&0&0&0\\0&0&0&0\end{array}\right)$
is no-signalling type and
$\left(\begin{array}{cc|cc}1&1&0&0\\1&1&0&0\\\hline0&0&0&0\\0&0&0&0\end{array}\right)$
is identity type).  If $B$ represents the Bell inequality
$\tr(B^T\Pi)\geq c$, then for reals $s$, $t$, $r$ with $r>0$, $M$
no-signalling type and $I$ identity type we have that
$\tilde{B}=rB+sM+tI$ represents the Bell inequality
$\tr(\tilde{B}^T\Pi)\geq rc+t$.  Although $\tilde{B}\neq B$, both
$\tilde{B}$ and $B$ are representations of the same Bell
inequality. For a Bell inequality of the form $\tr(B^T\Pi)\geq c$, any
local distribution, $\tilde{\Pi}$, for which $\tr(B^T\tilde{\Pi})=c$ is said to
\emph{saturate} the inequality.

Given a Bell inequality, we can construct others in the same scenario
by relabelling inputs and outputs.  In addition, if $n_A=n_B$ and
$m_A=m_B$ then we can also swap parties to construct others. Two Bell
inequalities related by such relabellings are said to be in the same
\emph{class}.  If $n_A=n_B$ and $m_A=m_B$ then there are
$2(n_A!)^{m_A}(n_B!)^{m_B}m_A!m_B!$ ways to relabel (and half as many
otherwise), although some relabellings may be equivalent to
others.\footnote{For example, when $(m_A,m_B,n_A,n_B)=(2,2,2,2)$,
  there are 32 relabellings, but applying each of these to a CHSH
  inequality generates only 8 unique inequalities.}  Because
relabellings do not change the essential features of a Bell
inequality, we focus on acquiring a representative of each Bell
inequality class, rather than the full list of inequalities.

Another important property of Bell inequalities is that they may be
``lifted'' to apply to scenarios with more inputs and/or
outputs~\cite{P1999}. We can add an input by setting the coefficients
corresponding to the new input to 0. (This means that the new Bell
inequality ignores the new input.) To increase the number of outputs,
the lifting involves copying the coefficients for one of the existing
outputs\footnote{The choice of output to copy may vary with each
  input. In addition we can also lift by only adding an output for one
  of the inputs, but not the others, e.g., to go from
  $[(3\ 4)\ (3\ 2)]$ to $[(3\ 4)\ (4\ 2)]$.}. This corresponds to
treating the new output in the same way as one of the existing
outputs.  This copying is needed to ensure that the new Bell
inequality has the same bound.  Note that these methods of lifting have
the property that the lifting of a facet Bell inequality always gives
a facet Bell inequality~\cite{P2005}.

To illustrate the concept of lifting, we present three Bell
inequalities: The first ($B_{\text{CHSH}}$) is the CHSH inequality for the
$(2,2,2,2)$ scenario. The second ($B_I$) lifts this to become a
$(2,3,2,2)$ inequality, whilst the third ($B_O$) is a lifting of $B_C$
to a $(2,2,2,3)$ inequality. All three are facet Bell inequalities
with the form $\tr(B^T\Pi)\geq 1$.
\begin{align*}
B_{\text{CHSH}}&=\left(\begin{array}{cc|cc}
1 & 0 & 1 & 0 \\
0 & 1 & 0 & 1 \\
\hline
1 & 0 & 0 & 1 \\
0 & 1 & 1 & 0
\end{array}\right)& 
B_I&=\left(\begin{array}{cc|cc|cc}
1 & 0 & 1 & 0 & 0 & 0 \\
0 & 1 & 0 & 1 & 0 & 0 \\
\hline
1 & 0 & 0 & 1 & 0 & 0 \\
0 & 1 & 1 & 0 & 0 & 0
\end{array}\right) &
B_O&=\left(\begin{array}{ccc|ccc}
1 & 0 & 0 & 1 & 0 & 0 \\
0 & 1 & 1 & 0 & 1 & 1 \\
\hline
1 & 0 & 0 & 0 & 1 & 1 \\
0 & 1 & 1 & 1 & 0 & 0
\end{array}\right)
\end{align*}

\subsection{The sets of no-signalling and quantum
  distributions}\label{allnosig}
For a given scenario, the set of all no-signalling distributions,
$\cNS_{(m_A,m_B,n_A,n_B)}$, forms a polytope whose facets correspond to
the positivity of probabilities.  The extremal no-signalling
distributions can in principle be found by vertex enumeration on these
facets.  In the general case, we do not know how to express all of
these, but local deterministic distributions are always vertices of
$\cNS_{(m_A,m_B,n_A,n_B)}$.  Furthermore, if both parties have only
two outputs per measurement, i.e., $n_A=n_B=2$, it is known that (up
to relabelling) all non-local extremal no-signalling distributions
have the form~\cite{JM2005}
\begin{equation}\label{eq:nsform}
   \begin{array}{lllllll}
     \begin{matrix}
       \qquad\qquad\ \ \   
       \overbrace{\rule{2.0cm}{0pt}}^{m_B-2-g}&\overbrace{\rule{1.6cm}{0pt}}^{g}
     \end{matrix}
     \\
    \begin{pmatrix}
       S & S & S &\ldots & S & L & \ldots & L \\
S & A & S/A & \ldots & S/A & L & \ldots & L \\
S & S/A & S/A & \ldots & S/A & L & \ldots & L \\
\vdots & \vdots & \vdots & & \vdots &\vdots & & \vdots\\
S & S/A & S/A & \ldots & S/A & L & \ldots & L \\
K & K & K & \ldots & K & M & \ldots & M \\
\vdots & \vdots & \vdots & & \vdots &\vdots & & \vdots\\
K & K & K & \ldots & K & M & \ldots & M
\end{pmatrix}
                                      \begin{array}{l}\phantom{\begin{matrix}0\\0\end{matrix}}\\\left.\phantom{\begin{matrix}0\\0\\0\end{matrix}}\right\rbrace{m_A-2-h}\\\left.\phantom{\begin{matrix}0\\0\\0\end{matrix}}\right\rbrace{h}
                                      \end{array}
   \end{array}
 \end{equation}
where $g\in\{0,1,\ldots,m_B-2\}$, $h\in\{0,1,\ldots,m_A-2\}$ and with the following $2\times 2$ blocks:
\begin{align*}
S=&\left(\begin{array}{cc}
\frac{1}{2} & 0\\
 0 & \frac{1}{2}
 \end{array}\right) & 
A= &\left(\begin{array}{cc}
 0 & \frac{1}{2}\\
 \frac{1}{2} & 0
 \end{array}\right) &
 K= &\left(\begin{array}{cc}
 \frac{1}{2} & \frac{1}{2}\\
 0 & 0
 \end{array}\right) &
L= &\left(\begin{array}{cc}
 \frac{1}{2} & 0  \\
 \frac{1}{2} & 0
 \end{array}\right) &
 M=&\left(\begin{array}{cc}
 1 & 0  \\
 0 & 0
 \end{array}\right).
\end{align*}

The set of quantum distributions is a subset of $\cNS$.  It is convex,
but not a polytope.  A distribution $P_{AB|XY}$ is quantum if there
exist POVMs $\{E_a^x\}_a$ and $\{F_b^y\}_b$ and a bipartite quantum
state $\rho_{AB}$ such that
$P_{AB|xy}(a,b)=\tr((E_a^x\ot F_b^y)\rho_{AB})$ for all $a$, $b$, $x$
and $y$.  We use $\cQ_{(m_A,m_B,n_A,n_B)}$ to denote the set of
quantum distributions in each scenario. Given a distribution
$P_{AB|XY}$ it is difficult to decide whether it is quantum.  To cope
with this a series of outer approximations to the quantum set were
introduced~\cite{NPA2008}. For positive integer $k$, we use
$\cQ^k_{(m_A,m_B,n_A,n_B)}$ to denote the set of correlations at the
$k^{\text{th}}$ level.  These levels form a hierarchy, in that
$$\cQ_{(m_A,m_B,n_A,n_B)}\subseteq\cQ^k_{(m_A,m_B,n_A,n_B)}\subseteq\cQ^l_{(m_A,m_B,n_A,n_B)}$$
for any positive integers $k>l$.  The advantage of using these sets is
that testing for membership of $\cQ^k_{(m_A,m_B,n_A,n_B)}$ is a
semidefinite program, which is in practice tractable for small enough
$k,m_A,m_B,n_A,n_B$.

Given a distribution $\Pi$, we use the following measure of its non-locality.
\begin{definition}[\cite{ZKBA1999}]\label{primal}
The \emph{local weight} of a distribution $\tilde{\Pi}$ is the solution to the problem:
\begin{align}
  \max_{\{x_i\}}\ \ &\sum_ix_i\nonumber\\
  \text{subject to }\ & \sum_ix_iP^{\rL,i}\leq \tilde{\Pi}\label{eq:prim}\\
  &x_i \geq 0\text{ for all }i\,,\nonumber
\end{align}
where $\sum_ix_iP^{\rL,i}\leq \tilde{\Pi}$ is interpreted
component-wise.\footnote{This condition ensures that
  $\tilde{\Pi}-\sum_ix_iP^{\rL,i}$ is equal to $(1-\sum_ix_i)\tilde{P}$ for
  $(1-\sum_ix_i)\geq0$ and with $\tilde{P}$ as a valid distribution.}
\end{definition}
If $\tilde{\Pi}$ is local then the local weight is $1$, while if $\tilde{\Pi}$ is a
non-local extremal no-signalling distribution, its local weight is
$0$.  Note that computing the local weight is a linear program.  Its dual
can be written
\begin{align}
  \min_M\ \ &\tr(M^T\tilde{\Pi}) \nonumber\\
  \text{subject to }\ &\tr(M^TP^{\rL,i})\geq1 \text{ for all }i\label{eq:dual}\\
          &M\geq 0\,,\nonumber
\end{align}
where again $M\geq 0$ is treated component-wise, and $i$ runs over all
local deterministic distributions.  Note that $\tr(M^TP^{\rL,i})\geq1$
for all $i$ implies that $\tr(M^T\Pi)\geq1$ is a (possibly trivial)
Bell inequality. If $\tilde{\Pi}$ is non-local, the matrix $M^*$ that
achieves the minimum is a non-trivial inequality violated by
$\tilde{\Pi}$.

\begin{definition}
  Let $\{x^*_i\}$ be the argument that achieves the optimum in the
  definition of the local weight of a distribution.  The \emph{local
    part} of a distribution $\tilde{\Pi}$ is $\sum_ix^*_iP^{\rL,i}$ and the
  \emph{non-local part} is $\tilde{\Pi}-\sum_ix^*_iP^{\rL,i}$.
\end{definition}
Note that the local and non-local parts of $\tilde{\Pi}$ are not in general
normalized, but can be easily renormalized, for instance by dividing
by the sum of all elements and multiplying by $m_Am_B$.

\section{Generating Facet Bell Inequalities}\label{sec:gen}
We can use the insight of the previous section to find Bell
inequalities by solving the dual problem~\eqref{eq:dual} for non-local
distributions.  Note that the Bell inequalities that emerge as
solutions have the form where all entries are positive and have local
bound $1$.  It turns out that all Bell inequalities can be represented
in such a form (see Lemma~\ref{posformlemma} in
Appendix~\ref{App:Proof}).  Furthermore, for every non-trivial facet
Bell inequality of this form there exists a non-local extremal
no-signalling distribution which gives value 0 for this Bell
expression (see Lemma~\ref{every} in Appendix~\ref{App:Proof}).  In
addition, there is a non-local extremal no-signalling distribution
achieving this and that takes the form~\eqref{eq:nsform} with $g=h=0$
(see Theorem~\ref{noKLM} in Appendix~\ref{App:Proof}).

This suggests that we could find all the Bell inequalities by running
the dual program for all non-local extremal no-signalling
distributions (in cases where these are known).  However, there is a
hidden subtlety: although 0 is the minimum possible value for any Bell
expression of this form (hence no other Bell inequality can have a
larger violation) there can be several Bell expressions all of which
have value 0 at the same extremal no-signalling distribution.  This
means that the output Bell inequality may not be a facet inequality,
and that some facet inequalities may be missed.  To mitigate this
problem, we can reduce the degeneracy by mixing extremal no-signalling
distributions with local distributions prior to using them as the dual
problem's objective function.  This is the idea behind our
algorithm. In principle one can choose enough local distributions to
break this degeneracy completely (see Appendix~\ref{App:Guar}); in
practice though we only mix with two local distributions to reduce the
degeneracy whilst keeping a reasonable runtime.

The idea to use the dual of the local weight problem to find Bell
inequalities was used before in~\cite{SBSL2016}, where a procedure
similar to our Algorithm~2 was used to find 18 Bell inequalities in
the $(3,3,3,3)$ scenario.

\subsection{A Linear Programming Algorithm for Bell Inequalities}
Our algorithm needs a few sub-algorithms.

\begin{enumerate}
\item\label{sub1} This is a way to decide whether a Bell inequality
  $B$ is a facet.  To do so, we find the set of local deterministic
  distributions for which $\tr(B^TP^{\rL,i})=1$ (i.e., those that
  achieve equality in the Bell inequality).  Call the values of $i$
  giving equality $a(1),a(2),\ldots, a(t)$. These all lie on a face,
  which is a facet if its dimension is one less than that of the
  entire space.  In other words we have to check how many dimensions
  are spanned by $\{P^{\rL,a(i)}-P^{\rL,a(1)}\}_{i=2}^t$.  If this is
  $(m_A(n_A-1)+1)(m_B(n_B-1)+1)-2$ then $B$ represents a facet Bell
  inequality.  (This dimensionality can be found by
  computing the rank of a matrix whose rows comprise the elements of
  $P^{\rL,a(i)}-P^{\rL,a(1)}$ for $i=2,\ldots,t$.)

\item\label{sub2} This checks whether two matrices $B$ and
  $\tilde{B}$ are representations of the same inequality.  To do so we
  compute two vectors, $v$ and $\tilde{v}$ with components
  $v_i=\tr(B^TP^{\rL,i})$ and $\tilde{v}_i=\tr(\tilde{B}^TP^{\rL,i})$.  By
  construction, the smallest element of each of these vectors is 1.
  Let the second smallest value of $v$ be $s>1$.  We perform an affine
  transformation that maps $1$ to $1$ and $s$ to $2$, (i.e., the
  function $x\mapsto\frac{1}{s-1}x+\frac{s-2}{s-1}$) to each component
  of $v$ forming $v'$.  A similar procedure is performed on
  $\tilde{v}$ forming $\tilde{v}'$.  The matrices $B$ and
  $\tilde{B}$ represent the same inequality if and only if
  $v'=\tilde{v}'$.

\item\label{sub3} Because we are interested in classes of Bell
  inequality, rather than the inequalities themselves, we also need to
  check whether two matrices are equivalent up to relabellings.  This
  algorithm checks whether $v'=T_m(\tilde{v}')$ where $m$ runs over
  all the relabellings, and $T_m$ is the permutation on the entries of
  $\tilde{v}'$ corresponding to the $m^{\text{th}}$ relabelling (a
  list of such permutations can be computed once before commencing the
  main algorithm to speed up this check, although for larger cases a
  lot of memory is required to store them all).  Note that if $v'$ and
  $\tilde{v}'$ do not have the same numbers of each type of entry (the
  same \emph{tally}) then there cannot be such a permutation.  We
  hence first check for this before running over all the permutations
  corresponding to relabellings.
\end{enumerate}

\subsubsection*{Algorithm 1}
This algorithm generates new facet inequalities for cases where
$n_A=n_B=2$.  It takes input $\eps\in(0,2/3)$, and a list $W$ of
known facet inequalities (which could be empty).
\begin{enumerate}
\item Set $j=1$.
\item\label{st:2} Set $\tilde{\Pi}$ to be the $j^{\text{th}}$ extremal
  no-signalling distribution of the form~\eqref{eq:nsform} with
  $g=h=0$.
\item\label{st1:3} Solve the dual problem~\eqref{eq:dual} for $\tilde{\Pi}$ using a simplex
  algorithm, giving the matrix $M$ that minimizes $\tr(M^T\tilde{\Pi})$.
\item Generate a list of values of $i$ such that
  $\tr(M^TP^{\rL,i})=1$.  Call these $a(1),a(2),\ldots,a(t)$.
\item Check whether $M$ defines a facet (using
  subalgorithm~\ref{sub1}) and whether it or a Bell
  inequality in the same class is in the list $W$ (using
  subalgorithms~\ref{sub2} and~\ref{sub3}).  If not, add $M$ to
  $W$.
\item\label{st:jk} Choose 2 distinct elements $k,l$ from $\{1,2,\ldots,t\}$ and
  form $\Pi'=(1-\frac{3\eps}{2})\tilde{\Pi}+\eps P^{\rL,a(k)}+\frac{\eps}{2}P^{\rL,a(l)}$.
\item\label{st1:7} Solve the dual problem~\eqref{eq:dual} for $\Pi'$ using a
  simplex algorithm, giving the matrix $M'$ that minimizes
  $\tr((M')^T\Pi')$.
\item\label{st:chk} Check whether $M'$ defines a facet  (using
  subalgorithm~\ref{sub1}) and whether it
  or a Bell inequality in the same class is in the list $W$ (using
  subalgorithms~\ref{sub2} and~\ref{sub3}).  If not,
  add $M'$ to $W$.
\item Repeat steps~\ref{st:jk}--\ref{st:chk} running over all distinct
  pairs $k,l$ from $\{1,2,\ldots,t\}$.
\item If $j< 2^{(m_A-1)(m_B-1)-1}$ set $j=j+1$ and return to
  step~\ref{st:2}, otherwise end the algorithm, outputting $W$.
\end{enumerate}

\subsubsection*{Algorithm 2}
In cases where the complete set of extremal no-signalling vertices is
not known, we use another algorithm to find facet Bell inequalities.
This works by picking random quantum-realisable distributions instead
of extremal no-signalling distributions.  This algorithm takes as
input a number of iterations, $j_{\max}$, and a list $W$ of known facet
inequalities (which could be empty).
\begin{enumerate}
\item Set $j=1$.
\item\label{st2:2} Randomly choose a pure quantum state of dimension
  $(\max(n_A,n_B))^2$ and $m_A$ random projective measurements of
  dimension $\max(n_A,n_B)$ with $n_A$ outcomes and $m_B$ random
  projective measurements of dimension $\max(n_A,n_B)$ with $n_B$
  outcomes.\footnote{For more details on how we do this, see
    Appendix~\ref{App:QDis}.}  Form the (quantum) distribution $\tilde{\Pi}$
  by computing the distribution that would be observed for this state
  and measurements.
\item\label{st2:3} Solve the dual problem~\eqref{eq:dual} for $\tilde{\Pi}$
  using a simplex algorithm, giving the matrix $M$ that minimizes
  $\tr(M^T\tilde{\Pi})$.
\item Check whether $M$ defines a facet (using
  subalgorithm~\ref{sub1}) and whether it or a Bell
  inequality in the same class is in the list $W$ (using
  subalgorithms~\ref{sub2} and~\ref{sub3}).  If not, add $M$ to
  $W$.
\item If $j< j_{\max}$ set $j=j+1$ and return to
  step~\ref{st2:2}, otherwise end the algorithm, outputting $W$.
\end{enumerate}

We also consider Algorithm~$2'$, which is the same as Algorithm~2
except that the following additional step is added between
Steps~\ref{st2:2} and~\ref{st2:3}.\bigskip

\ \ 2b. If $\tilde{\Pi}$ is non-local, replace $\tilde{\Pi}$ with the
renormalized non-local part of $\tilde{\Pi}$.

\subsection{Comments on the algorithms}\label{sub:comm}
We use the simplex algorithm as implemented by Mathematica~11.1.1.
Unfortunately, full details of this specific implementation are not
publicly available (to our knowledge).  In particular, for our
problem, in spite of the steps taken to break some of the degeneracy,
for a given $\Pi'$ there remain many $M'$ that achieve the optimum for
the dual.  Which one is given out by the simplex algorithm depends on
the details of how it decides which edge to travel along when faced
with several possibilities. Mathematica's implementation is
deterministic, but for fixed $k$ and $l$, small changes in $\eps$ can
lead to a different solution based on $\Pi'$. It is hence useful to
rerun the algorithm for several values of $\eps$.

As mentioned above, one disadvantage of the above algorithm is that in
many cases the output of the dual program does not correspond to a
facet inequality. It is possible to alter the problem such that the
solution space is the polar dual of the local polytope, where indeed
every solution will be facet. However, by doing this, it is no longer
the case that every facet inequality will be given as a solution,
regardless of the input no-signalling distribution (i.e., the analogue
of Lemma~\ref{every} does not hold). This is further discussed in
Appendix~\ref{App:Polar}.  An alternative way to find facet Bell
inequalities from lower dimensional faces was used in~\cite{BGP2010}.

Running through all permutations to check whether two Bell matrices
are in the same class can take time, so Algorithm~1 can also be run
with a modified subalgorithm~\ref{sub3} in which $M$ is added to $W$
if the vectors $v'$ and $\tilde{v}'$ have different tallies.  Running
the algorithm in this way can generate many classes of Bell
inequality, but without also running through all permutations, some
classes may be missed.  For instance, in the case
$(m_A,m_B,n_A,n_B)=(4,4,2,2)$, for $i=1,2$, the Bell inequalities
$\tr(B_i^T\Pi)\geq1$ where
\begin{equation*}
B_1=\left(
\begin{array}{cc|cc|cc|cc}
 0 & \frac{1}{2} & 0 & 0 & 0 & 0 & 0 & 0 \\
 0 & 0 & \frac{1}{2}\vphantom{\frac{1}{f}} & 0 & 0 & 0 & \frac{1}{2} & 0 \\
 \hline
 0 & 0 & 0 & 0 & 0 & \frac{1}{2} & 0 & 0 \\
 \frac{1}{2} & 0 & 0 & \frac{1}{2}\vphantom{\frac{1}{f}} & 0 & 0 & \frac{1}{2} & 0 \\
  \hline
 0 & 0 & 0 & \frac{1}{2} & 0 & 0 & \frac{1}{2} & 0 \\
 \frac{1}{2} & 0 & 0 & 0 & \frac{1}{2}\vphantom{\frac{1}{f}} & 0 & 0 & 0 \\
  \hline
 0 & \frac{1}{2} & \frac{1}{2} & 0 & \frac{1}{2} & 0 & 0 & 0 \\
 \frac{1}{2} & 0 & 0 & 0 & 0 & 0 & 0 & \frac{1}{2}\vphantom{\frac{1}{f}} \\
\end{array}
\right)\qquad\text{and}\qquad
B_2=\left(
\begin{array}{cc|cc|cc|cc}
 0 & \frac{1}{2} & 0 & \frac{1}{2} & 0 & 0 & 0 & 0 \\
 0 & 0 & 0 & 0 & 0 & 0 & \frac{1}{2}\vphantom{\frac{1}{f}} & 0 \\
  \hline
 0 & \frac{1}{2} & 0 & 0 & 0 & 0 & 0 & \frac{1}{2} \\
 0 & 0 & 0 & \frac{1}{2}\vphantom{\frac{1}{f}} & \frac{1}{2} & 0 & 0 & 0 \\
  \hline
 0 & 0 & 0 & 0 & 0 & \frac{1}{2} & 1 & 0 \\
 \frac{1}{2} & 0 & \frac{1}{2}\vphantom{\frac{1}{f}} & 0 & 0 & 0 & 0 & 0 \\
  \hline
 0 & \frac{1}{2} & \frac{1}{2} & 0 & \frac{1}{2} & 0 & 0 & 0 \\
 0 & 0 & 0 & 0 & 0 & 0 & 0 & \frac{1}{2}\vphantom{\frac{1}{f}} \\
\end{array}
\right) 
\end{equation*}
are in different classes, but the corresponding vectors $v'$ and
$\tilde{v}'$ have the same tallies.

Another disadvantage of our algorithm is that we do not have a
criteria for deciding when the list of classes found is complete. In
cases where we are sure we have found all Bell inequalities, we know
this only because the total number had already been found by other
means.

Note that because Algorithm~2 is based on quantum distributions, it is
unable to find Bell inequalities for which there is no quantum
violation~\cite{ABBAGP2010}.  Furthermore, the chosen quantum
distribution may be local, in which case the dual program will not
output a Bell inequality.  To circumvent these disadvantages, other
ways to pick non-local distributions can be used.  For instance, the
renormalized non-local part of a distribution may violate Bell
inequalities that the original distribution does not. This motivates
our Algorithm~$2'$, which may be able to find Bell inequalities
without quantum violation (see also Remark~\ref{rmk:2}).  If there
were an extension of the work of~\cite{JM2005} to cases with $n_A>2$
or $n_B>2$, we could use Algorithm~1 in these cases.  We expect that
this would be a quicker way to generate new Bell inequalities.

\section{Results Using the Algorithm}\label{sec:results}
In this section we summarise the results we have obtained using the
above algorithm.  Due to the large number of inequality classes found
we present them in separate files available at~\cite{web_link}, along
with a file explaining how to import and use them in both Mathematica
and Matlab.  We give both the version the algorithm found (the ``raw''
version) and a second version after an affine transformation analogous
to that mentioned in subalgorithm 2 has been applied (the ``affine''
version). These are presented after relabellings that make obvious
input/output liftings.

\subsection{$(3,5,2,2)$ scenario}\label{3522Sec}
The number of facet classes for this scenario (7) was given
in~\cite{DS2015}. Six of these are also facet classes of the $(3,4,2,2)$
scenario, and are given in~\cite{CG2004}. The remaining class was first found in \cite{QVB2014}. We used this case to test Algorithm~1, rederiving the result. We found a representative of the new class to be $\tr(I_{3522}^T\Pi)\geq 1$
\begin{equation}\label{eq:I3522}
I_{3522}=\left(\begin{array}{cc|cc|cc|cc|cc}
 0 & \frac{2}{3} & 0 & 0 & 0 & \frac{1}{3} & 0 & 0 & 0 & \frac{1}{3} \\
 0 & 0 & \frac{2}{3} & 0 & \frac{1}{3}\vphantom{\frac{1}{f}} & 0 & 0 & 0 & \frac{1}{3} & 0 \\
 \hline
 0 & 0 & 0 & 0 & \frac{2}{3} & 0 & \frac{2}{3} & 0 & 0 & 0 \\
 \frac{2}{3} & 0 & 0 & \frac{2}{3}\vphantom{\frac{1}{f}} & 0 & 0 & 0 & 0 & 0 & 0 \\
 \hline
 0 & 0 & 0 & \frac{2}{3} & 0 & \frac{1}{3} & 0 & 0 & \frac{1}{3} & 0 \\
 0 & 0 & 0 & 0 & \frac{1}{3} & 0 & 0 & \frac{2}{3}\vphantom{\frac{1}{f}} & 0 & \frac{1}{3} \\
\end{array}
\right)\,.
\end{equation}

\subsection{$(4,4,2,2)$ scenario}\label{4422Sec}
For this scenario, we have enumerated all $175$ inequivalent classes
of facet Bell inequality (including the trivial positivity
inequality). Whilst the number of classes was already
known~\cite{DS2015}, a complete list was not provided. A partial list
of 129 non-trivial inequalities was given in~\cite{PV2009}. Our
generation of these inequalities was performed on a standard desktop
computer and took a few days.

In order to give an idea of the symmetries of this polytope,
Table~\ref{facets4422} gives the number of members of each class.
\begin{table}
\begin{center}
\begin{tabular}{c|c|c}
Size of Class & Number of Classes & Notable cases \\
\hline
64 & 1 & Positivity \\
288 & 1 & CHSH~\cite{CHSH1969} \\
9216 & 2 & $I_{3322}$~\cite{CG2004} \\
18432 & 4 & $I_{4422}$~\cite{CG2004} \\
24576 & 1 & \\
36864 & 4 & \\
49152 & 2 & \\
73728 & 8 & \\
98304 & 2 & \\
147456 & 61 &\\
294912 & 89 & \\
\hline
36391264  & 175
\end{tabular}
\caption[Facet Bell inequality analysis for the $(4,4,2,2)$
scenario]{The size of each facet class for the $(4,4,2,2)$ local
  polytope. $294912$ is the size of the relabelling symmetry
  group.}\label{facets4422}
\end{center}
\end{table}

\subsection{$(2,3,3,2)$ scenario}\label{2332Sec}
The extremal no-signalling distributions are not known for this
scenario, but using Algorithm~$2'$ we found five classes
of Bell inequality: the positivity condition, a lifting of CHSH and
three new inequality classes $\tr((I_{2332}^1)^T\Pi)\geq 1$ and
$\tr((I_{2332}^2)^T\Pi)\geq 1$ and $\tr((I_{2332}^3)^T\Pi)\geq 1$ with
representative matrices
\begin{align*}
I^1_{2332}&=\left(
\begin{array}{cc|cc|cc}
 0 & \frac{1}{2} & 0 & \frac{1}{2} & 1 & 0 \\
 0 & \frac{1}{2} & 1 & 0 & 0 & \frac{1}{2} \\
 1 & 0 & 0 & \frac{1}{2} & 0 & \frac{1}{2}\vphantom{\frac{1}{f}} \\
 \hline
 \frac{1}{2} & 0 & \frac{1}{2} & 0 & \frac{1}{2} & 0 \\
 0 & 1 & 0 & 1 & 0 & 1 \\
 0 & 1 & 0 & 1 & 0 & 1
\end{array}
\right) & I^2_{2332}&=\left(
\begin{array}{cc|cc|cc}
 0 & \frac{1}{2} & 0 & 1 & \frac{1}{2} & 0 \\
 1 & 0 & \frac{1}{2} & 0 & \frac{1}{2} & 0 \\
 0 & \frac{1}{2} & \frac{1}{2}\vphantom{\frac{1}{f}} & 0 & 0 & 1 \\
 \hline
 0 & 1 & 1 & 0 & 1 & 0 \\
 \frac{1}{2} & 0 & 0 & \frac{1}{2} & 0 & \frac{1}{2} \\
 \frac{1}{2} & 0 & 0 & \frac{1}{2} & 0 & \frac{1}{2}
\end{array}
\right)
& I^3_{2332}&=\left(
\begin{array}{cc|cc|cc}
 1 & 1 & 0 & 0 & 0 & 0 \\
 1 & 0 & 0 & 1 & 0 & 1 \\
 1 & 0 & 1 & 0 & 1 & 0\vphantom{\frac{1}{f}} \\
 \hline
 0 & 1 & 0 & 0 & 0 & 0 \\
 0 & 0 & 0 & 1 & 1 & 0 \\
 0 & 0 & 1 & 0 & 0 & 1 
\end{array}
\right).
\end{align*}
For this scenario, we were able to perform facet enumeration of the local polytope, and verify that these five form the complete list of classes. Table~\ref{2332facets} gives the number of inequalities in each class.
\begin{table}
\begin{center}
\begin{tabular}{cc}
Bell Inequality Class & Number of Faces \\
\hline
Positivity & 36 \\
CHSH & 216 \\
$I^1_{2332}$ & 288 \\
$I^2_{2332}$ & 288 \\
$I^3_{2332}$ & 432 \\
\hline
Total & 1260
\end{tabular}
\caption{The facets of the $(2,3,3,2)$ local polytope, sorted into their inequality classes.}
\label{2332facets}
\end{center}
\end{table}

\subsection{$(3,3,2,3)$ scenario}\label{3323Sec}
In this scenario we again do not know the full list of extremal
no-signalling distributions, and we employed Algorithm $2$ for
$250000$ runs, finding 19 inequality classes. We then employed
Algorithm~$2'$ for $3451207$ runs to find 6 more classes (a total of
25). As no new class was found for over $1164000$ runs, we speculate
this is the full list of classes.  We also used Algorithm~4 (detailed
in Appendix~\ref{App:BellDim}) to find the smallest scenario for which
each of the inequalities first appears.  In Table~\ref{3332classes} we
present this analysis. Surprisingly, only one of the classes found is
not a lifting from a smaller scenario.
\begin{table}
\begin{center}
\begin{tabular}{ccccccc}
Pre-Lifting Scenario && Number of Classes Found && Number of Distinct Classes in Pre-Lifting Scenario && Notable Inequality\\
\hline
$[(2)\ (2)]$ && 1 && 1 && Positivity Condition \\
$[(2\ 2)\ (2\ 2)]$ && 1 && 1 && CHSH \\
$[(2\ 2)\ (2\ 3)]$ && $-$ && $-$ && \\
$[(2\ 2)\ (3\ 3)]$ && $-$ && $-$ && \\
$[(2\ 2)\ (2\ 2\ 2)]$ && $-$ && $-$ &&  \\
$[(2\ 2)\ (2\ 2\ 3)]$ && $-$ && $-$ &&  \\
$[(2\ 2)\ (2\ 3\ 3)]$ && $-$ && $-$ &&  \\
$[(2\ 2)\ (3\ 3\ 3)]$ && $-$ && $-$ &&  \\
$[(2\ 2\ 2)\ (2\ 3)]$ && 2 && 1 &&  \\
$[(2\ 2\ 2)\ (3\ 3)]$ && 1 && 1 &&  \\
$[(2\ 2\ 2)\ (2\ 2\ 2)]$ && 3 && 1 && $I_{3322}$  \\
$[(2\ 2\ 2)\ (2\ 2\ 3)]$ && 8 && 3 &&   \\
$[(2\ 2\ 2)\ (2\ 3\ 3)]$ && 8 && 4 &&   \\
$[(2\ 2\ 2)\ (3\ 3\ 3)]$ && 1 && 1 &&   \\
\end{tabular}
\caption{The inequalities classes found for $(3,3,2,3)$, sorted to
  show their degree of lifting. Relabelling of inputs has been
  used to combine scenarios of equivalent dimension e.g.,
  $[(2\ 2\ 2)\ (2\ 2\ 3)]$ and $[(2\ 2\ 2)\ (2\ 3\ 2)]$. We use ``$-$'' to denote
  cases for which it is known that no new classes exist. Note that
  classes that are unique before lifting may define several
  inequivalent classes in a higher dimension. This is why although
  $I_{3322}$ is the only new non-trivial inequality class in the $[(2\ 2\ 2)\ (2\ 2\ 2)]$
  scenario, we found three inequivalent classes corresponding to
  $I_{3322}$ liftings.}
\label{3332classes}
\end{center}
\end{table}

\subsection{$(4,5,2,2)$ scenario}\label{4522Sec}
There is no known total of inequivalent classes for this
scenario. Using Algorithm~1 we have established the existence of at
least $18277$ classes and have a representative of each.  We expect
the total number of classes to be significantly
larger\footnote{\label{ft:}We stopped the algorithm after a few weeks
  but new inequality classes were being found regularly.}. As for this
case the size of the relabelling symmetry group is quite large, which
slows down the check to see if a previous class member has been found,
we use the modified version of subalgorithm~\ref{sub3} discussed in
Section~\ref{sub:comm} in which only the tally of $v'$ is
checked---this reduces the runtime considerably, but means that we may
discard new classes that have the same tally to a class that
has already been found.

\subsection{$(3,3,3,3)$ scenario}\label{3333Sec}
Again for this scenario we do not know all the extremal no-signalling
distributions, so we use Algorithm~2 to find facet inequalities and
the modified version of subalgorithm~\ref{sub3} discussed in
Section~\ref{sub:comm} in which only the tally of $v'$ is checked. In
our initial run, we found $10143$ classes with
$j_{\max}=2.5$million. We then ran Algorithm~$2'$ with
$j_{\max}=0.6$million (again with the modified version of
subalgorithm~\ref{sub3}), and found a further $11018$ classes. These
figures suggest $2'$ is significantly better at finding new classes.
We did a small further check to confirm that there are inequivalent
inequalities with the same tallies generating 9
further inequalities.  We expect the total number of classes to be
significantly larger than the figure presented here (see
Footnote~\ref{ft:}).

Our list of inequalities can be augmented with the results
of~\cite{SBSL2016}, who present $18$ inequalities for this scenario,
of which $3$ are contained
in our list.  Again we know that, of these classes, some of them
will correspond to liftings of lower dimensional
inequalities. Table~\ref{33scenarios} partitions our 
$21170$ classes into the scenario for which they first appear, taking
into account cases in which different inputs have different numbers of
outputs.  This uses Algorithm~4, detailed in
Appendix~\ref{App:BellDim}.

\begin{table}
\begin{center}
\begin{tabular}{ccccccc}
Pre-Lifting Scenario && Number of Classes Found && Number of Distinct Classes in Pre-Lifting Scenario && Notable Inequality \\
\hline
$[(2)\ (2)]$ && 1 && 1 && Positivity Condition \\
$[(2\ 2)\ (2\ 2)]$ && 2 && 1 && CHSH \\
$[(2\ 2)\ (2\ 3)]$ && $-$ && $-$ && \\
$[(2\ 2)\ (3\ 3)]$ && $-$ && $-$ &&\\
$[(2\ 3)\ (2\ 3)]$ && $-$ && $-$ &&\\
$[(2\ 3)\ (3\ 3)]$ && $-$ && $-$ &&\\
$[(3\ 3)\ (3\ 3)]$ && 1 && 1 && $I_{2233}$\\
$[(2\ 2)\ (2\ 2\ 2)]$ && $-$ && $-$ && \\
$[(2\ 2)\ (2\ 2\ 3)]$ && $-$ && $-$ && \\
$[(2\ 2)\ (2\ 3\ 3)]$ && $-$ && $-$ && \\
$[(2\ 2)\ (3\ 3\ 3)]$ && $-$ && $-$ && \\
$[(2\ 3)\ (2\ 2\ 2)]$ && 7 && 1  \\
$[(2\ 3)\ (2\ 2\ 3)]$ && 0 && 0 && \\
$[(2\ 3)\ (2\ 3\ 3)]$ && 0 && 0 && \\
$[(2\ 3)\ (3\ 3\ 3)]$ && 0 && 0 && \\
$[(3\ 3)\ (2\ 2\ 2)]$ && 2 && 1 && \\
$[(3\ 3)\ (2\ 2\ 3)]$ && 5 && 3 && \\
$[(3\ 3)\ (2\ 3\ 3)]$ && 9 && 7 && \\
$[(3\ 3)\ (3\ 3\ 3)]$ && 4 && 4 && \\
$[(2\ 2\ 2)\ (2\ 2\ 2)]$ && 9 && 1 && $I_{3322}$ \\
$[(2\ 2\ 2)\ (2\ 2\ 3)]$ && 27 && 3 && \\
$[(2\ 2\ 2)\ (2\ 3\ 3)]$ && 15 && 4 && \\
$[(2\ 2\ 2)\ (3\ 3\ 3)]$ && 3 && 1 && \\
$[(2\ 2\ 3)\ (2\ 2\ 3)]$ && 154 && 40 && \\
$[(2\ 2\ 3)\ (2\ 3\ 3)]$ && 762 && 337 && \\
$[(2\ 2\ 3)\ (3\ 3\ 3)]$ && 398 && 276 && \\
$[(2\ 3\ 3)\ (2\ 3\ 3)]$ && 2532 && 1764 && \\
$[(2\ 3\ 3)\ (3\ 3\ 3)]$ && 6637 && 6060 && \\
$[(3\ 3\ 3)\ (3\ 3\ 3)]$ && 10602 && 10602 && \\
\end{tabular}
\caption{The inequalities classes found for $(3,3,3,3)$, sorted to
  show their degree of lifting. Due to the additional symmetry, Relabelling of parties/inputs has been used to group classes, e.g. $[(2\ 3)\ (2\ 2\ 2)]$ and $[(2\ 2\ 2)\ (2\ 3)]$. Again we find that unique classes in lower dimensions can be lifted to define several facet classes.}\label{33scenarios}
\end{center}
\end{table}

\section{Application to the detection loophole}\label{sec:det}
In the remainder of this article we investigate whether these
inequalities allow us to lower the efficiency required for closing the detection
loophole.

To perform a bipartite Bell experiment, entangled photons are sent to
two detectors where they are measured.  After repeating many times
using different randomly chosen measurements we can build up an
estimate of the distribution $\Pi$, from which we can see whether a
particular Bell inequality is violated or not.  One of the issues with
such an experiment is that real detectors sometimes fail to detect.
The question is then how to certify that the setup is non-local in the
presence of such no-click events.  In particular, we would like to
know the minimal detection efficiency at which we can still certify
the presence of non-locality. Given that certifying the presence of
non-locality is necessary for device-independent tasks, it is
important to be able to treat cases with imperfect detectors.

To model this, we assume that each detector has an efficiency,
representing the probability that it will click when it should.  For
simplicity, we consider the case where the efficiencies of each
detector are the same and call this parameter $\eta$.  No-click can be
treated as an additional outcome for each measurement that occurs with
probability $1-\eta$ (independently for each measurement).  Given a
probability distribution $P_{AB|XY}$, we use $P^{\eta}_{AB|XY}$ to
denote the inefficient detector version with efficiency $\eta$, which
is formed by adding the possible outcome ``$N$'' to each measurement
and taking
\begin{align*}
P^{\eta}_{AB|xy}(a,b)&=\eta^2P_{AB|xy}(a,b),\\
P^{\eta}_{AB|xy}(N,b)&=\eta(1-\eta)P_{B|y}(b),\\
P^{\eta}_{AB|xy}(a,N)&=\eta(1-\eta)P_{A|x}(a),\\
P^{\eta}_{AB|xy}(N,N)&=(1-\eta)^2,
\end{align*}
for all $a\in\{1,2,\ldots,n_A\}$, $b\in\{1,2,\ldots,n_B\}$,
$x\in\{1,2,\ldots,m_A\}$ and $y\in\{1,2,\ldots,m_B\}$.  We use
$\Pi^{\eta}$ to denote the matrix representation of
$P^{\eta}_{AB|XY}$.  (For a discussion of other ways to deal with
no-click events, see~\cite{Branciard11}.)

It is worth noting that for a given $\tilde{\Pi}\in\cNS_{(m_A,m_B,n_A,n_B)}$
we have that $\tilde{\Pi}^\eta\notin\cL_{(m_A,m_B,n_A+1,n_B+1)}$ only if
$\tilde{\Pi}\notin\cL_{(m_A,m_B,n_A,n_B)}$, and that for $\eta_1\geq \eta_2$, $\tilde{\Pi}^{\eta_2}\notin\cL_{(m_A,m_B,n_A+1,n_B+1)}$ implies
$\tilde{\Pi}^{\eta_1}\notin\cL_{(m_A,m_B,n_A+1,n_B+1)}$.

We now define the detection threshold for a given Bell inequality, and
for a given scenario.
\begin{definition}
  Given a $m_A(n_A+1)\times m_B(n_B+1)$ matrix $B$ such that
  $\tr(B^T\Pi)\geq c$ for all $\Pi\in\cL_{(m_A,m_B,n_A+1,n_B+1)}$, the \emph{detection
    threshold for $B$}, is defined by
\begin{equation}
\eta_B:=\inf\{\eta\in[0,1]:\exists\ \tilde{\Pi}\in\cQ_{(m_A,m_B,n_A,n_B)},\; \tr(B^T\tilde{\Pi}^\eta)<c\}\,.
\end{equation}
\end{definition}
This is the smallest detection efficiency such that $B$ can certify
non-locality using quantum states and measurements for all higher
efficiencies.  Note that $B$ is a Bell inequality for the
$(m_A,m_B,n_A+1,n_B+1)$ scenario.  Some Bell inequalities of this type
can be formed from those for the $(m_A,m_B,n_A,n_B)$ scenario by
lifting, as discussed in Section~\ref{probandbell}.
\begin{definition}
The \emph{detection threshold for the $(m_A,m_B,n_A,n_B)$ scenario} is
defined by
\begin{equation}
\eta_{(m_A,m_B,n_A,n_B)}:=\inf\{\eta\in[0,1]:\exists\ \tilde{\Pi}\in\cQ_{(m_A,m_B,n_A,n_B)},\; \tilde{\Pi}^\eta\notin\cL_{(m_A,m_B,n_A+1,n_B+1)}\}\,.
\end{equation}
\end{definition}

We can also define a detection threshold for a set
$\cS_{(m_A,m_B,n_A,n_B)}\subseteq\cNS_{(m_A,m_B,n_A,n_B)}$.
\begin{definition}
The \emph{detection threshold for $\cS_{(m_A,m_B,n_A,n_B)}$} is
defined by
\begin{equation}
\eta^\cS_{(m_A,m_B,n_A,n_B)}:=\inf\{\eta\in[0,1]:\exists\ \tilde{\Pi}\in\cS_{(m_A,m_B,n_A,n_B)},\; \tilde{\Pi}^\eta\notin\cL_{(m_A,m_B,n_A+1,n_B+1)}\}\,.
\end{equation}
\end{definition}

We will use the fact that if $\cQ$ is a subset of $\cS$ then the
detection threshold for $\cS$ is lower than the quantum one, i.e.,
$\cQ_{(m_A,m_B,n_A,n_B)}\subseteq\cS_{(m_A,m_B,n_A,n_B)}$ implies
$\eta^\cS_{(m_A,m_B,n_A,n_B)}\leq\eta_{(m_A,m_B,n_A,n_B)}$.

It is known that there exist states and measurements for which this
threshold tends to 0 as ${d\rightarrow\infty}$~\cite{M2002}. However,
to achieve this, the number of measurements required scales as
$O(2^d)$.  For practical reasons we are interested in cases with small
numbers of inputs and outputs.

In the $(2,2,2,2)$ scenario, taking $\tilde{\Pi}$ to be quantum correlations
that achieve Tsirelson's bound~\cite{C1980} for the CHSH inequality,
($\tr(B^T_{\text{CHSH}}\tilde{\Pi})=2-\sqrt{2}$) it has been shown that
$\tilde{\Pi}^\eta$ is non-local if and only if
$\eta>2(\sqrt{2}-1)\approx 82.8\%$~\cite{M1986}.  A lower value was
found by Eberhard~\cite{E1993}, who showed that one can use a
two-qubit state of the form $\cos\theta\ket{00}+\sin\theta\ket{11}$
and appropriate 2-outcome measurements to give rise to a distribution
$\hat{\Pi}_\theta$ such that for any $\eta>2/3$ there exists $\theta$ such
that $\hat{\Pi}^\eta_\theta$ is non-local, while for $\eta=2/3$ the
distribution $\hat{\Pi}^\eta_\theta$ is local for all $\theta$. Somewhat
counterintuitively, the state demonstrating non-locality has
$\theta\to0$ as $\eta\to2/3$.  In~\cite{VPB2010}, a $(4,4,2,2)$
inequality (which we refer to as $I_{4422}$~\cite{CG2004}) was
considered and a state and measurements on a four dimensional Hilbert
space were given demonstrating that
$\eta_{I_{4422}}\leq(\sqrt{5}-1)/2\approx 61.8\%$.  The state used has
the form
$\sqrt{(1-\epsilon^2)/3}\left(\ket{00}+\ket{11}+\ket{22}\right)+\epsilon\ket{33}$,
with the value $(\sqrt{5}-1)/2$ achieved in the limit
$\epsilon\rightarrow 0$.

In~\cite{MP2003}, the problem was abstracted away from specific Bell
inequalities, instead giving an explicit local-hidden variable
construction which can replicate any inefficient no-signalling
distribution, provided the detection efficiency is below
$(m_A+m_B-2)/(m_Am_B-1)$.  This hence corresponds to a lower bound on
$\eta_{(m_A,m_B,n_A,n_B)}$.  In the next subsection we improve on
these lower bounds in cases where $n_A=n_B=2$.

\subsection{A Fundamental Lower Bound on the Detection Threshold}\label{sec:NSLB}
In this section we show how to bound $\eta_{(m_A,m_B,n_A,n_B)}$ using
knowledge of the set of no-signalling distributions.  As discussed
above, since
$\cQ_{(m_A,m_B,n_A,n_B)}\subseteq\cNS_{(m_A,m_B,n_A,n_B)}$, a lower
bound for the detection threshold for all no-signalling distributions
will apply to the quantum case too.

In cases where we have a complete set of extremal no-signalling
distributions, we can obtain an arbitrarily good estimate
$\eta^\cNS_{(m_A,m_B,n_A,n_B)}$ using the following algorithm.

\subsubsection*{Algorithm~3}
The algorithm takes input $\delta\in(0,1)$, the tolerance we look for
in the solutions.
\begin{enumerate}
\item Set $j=1$, $\eta_c=1$ and $j_{\max}$ to be the total number of
  extremal no-signalling distributions.
\item \label{st3:2} Set $\tilde{\Pi}$ to be the $j^{\text{th}}$ non-local extremal
  no-signalling distribution.
\item Set $\eta_{\min}=0$ and $\eta_{\max}=1$.
\item\label{st3:4} Set $\eta'=(\eta_{\min}+\eta_{\max})/2$ and generate
  $\tilde{\Pi}_{\eta'}$.
\item Find the local weight of $\tilde{\Pi}_{\eta'}$ by solving the linear
  program~\eqref{eq:prim}, setting this to $w$.
\item If $w=1$, set $\eta_{\min}=\eta'$, otherwise, set
  $\eta_{\max}=\eta'$.
\item If $\eta_{\max}-\eta_{\min}>\delta$, go to step~\ref{st3:4}.
\item If $(\eta_{\min}+\eta_{\max})/2<\eta_c$, set
  $\eta_c=(\eta_{\min}+\eta_{\max})/2$.
\item If $j<j_{\max}$, set $j=j+1$ and return to step~\ref{st3:2}
  otherwise the algorithm ends, outputting $\eta_c$.
\end{enumerate}

This algorithm runs through all the non-local extremal no-signalling
distributions, and computes at what $\eta$ they become local.  Then by
taking the minimum over all such distributions we obtain the
no-signalling detection threshold for this scenario.

We have applied this algorithm to the $(m_A,m_B,2,2)$ scenario for
various $m_A$ and $m_B$.  The results are shown in
Table~\ref{fundamentalbounds}, where we have also compared with the
lower bound from~\cite{MP2003} in order to highlight the improvement
we obtain.

\begin{table}
\begin{center}
\begin{tabular}{cc|ccccc|c|c|}
&$m_A$&2&3&4&5&6&$2/(m_B+1)$\\
$m_B$&&&&&&&\\\hline
2&&2/3&2/3&2/3&2/3&2/3&2/3\\
3&&&4/7&5/9&5/9&5/9&1/2\\
4&&&&1/2&1/2&1/2&2/5\\
5&&&&&4/9& * &1/3\\
\end{tabular}
\caption[Fundamental bounds on the detection threshold]{The maximum
  detection efficiency such that $\Pi^\eta$ can be generated
  classically for any $\Pi\in\cNS_{(m_A,m_B,2,2)}$. The * in the
  $(5,6,2,2)$ case indicates that this was not evaluated due to the
  high number of non-local extremal no-signalling distributions. The
  final column is the lower bound of~\cite{MP2003}.  The main part of
  the table is populated with fractions, even though the
  algorithm outputs a decimal.  Strictly we should say that the
  value is consistent with the stated fraction to 7 decimal
  places.}\label{fundamentalbounds}
\label{tab:1}
\end{center}
\end{table}

\subsection{Bounding Detection Thresholds using the Semidefinite Hierarchy}
We can think of these no-signalling values as lower bounds on the
quantum detection thresholds.  When interpreted this way, we do not
expect these lower bounds to be tight, although, somewhat
surprisingly, they are in the case where $m_B=2$~\cite{E1993}.  In
order to give better bounds, we can use other supersets of the quantum
set, for instance, those based on the semidefinite
hierarchy~\cite{NPA2008}.  In other words, we can try to find
$\eta^{\cQ^k}_{(m_A,m_B,n_A,n_B)}$ for some level $k\in\mathbb{N}$ of the
hierarchy.

Since $\cQ^k_{(m_A,m_B,n_A,n_B)}$ is not a polytope, we cannot
directly use the method of Section~\ref{sec:NSLB}.  Instead, for each
value $\eta$ we perform a semidefinite optimisation over
$\cQ^k_{(m_A,m_B,n_A,n_B)}$, minimising $\tr(B^T\hat{\Pi}^{\eta})$ over
$\hat{\Pi}\in\cQ^k_{(m_A,m_B,n_A,n_B)}$ for a specific Bell inequality $\tr(B^T\Pi)\geq 1$
for the $(m_A,m_B,n_A+1,n_B+1)$ scenario. If the minimum is at least
1, we can conclude that there is no $\tilde{\Pi}\in\cQ_{(m_A,m_B,n_A,n_B)}$
for which $\tr(B^T\tilde{\Pi}^\eta)<1$. Performing the
computation for higher levels $k$ of the hierarchy gives successively
tighter bounds.

\subsection{Results}
All of the results in this section were obtained using the convex
optimisation interface CVX~\cite{cvx} for Matlab, with
MOSEK~\cite{mosek} as the solver. Unless stated otherwise, tests were
run at the default CVX precision. Note also that in the code, a value
of $\tr(B^T\Pi^\eta)\geq 1-\epsilon$ was considered local for
$\epsilon\approx 1.49\times 10^{-8}$. This mitigates against the
possibility of concluding a distribution is non-locality because of
the solver precision, when in fact it is not.  However, it can mean
that we sometimes incorrectly conclude a distribution is local.  Thus,
the values we obtain will be upper bounds of the threshold over the
considered set.

In the supplementary files these results are given in the following
format: first the inequality considered is given, followed by the
lifting choices of Alice, then the lifting choices of Bob, followed by
the threshold found.

\subsubsection{$(4,4,2,2)$ Scenario}
For this scenario an explicit quantum construction is given
in~\cite{VPB2010} achieving $(\sqrt{5}-1)/2\approx 0.6180$ ($61.80\%$) requiring
shared entangled states with local dimension 4. We were able to test
every inequality in this scenario at level $2$, with no inequality
beating $I_{4422}$, which was the inequality used
in~\cite{VPB2010}. We were only able to obtain the value $61.83\%$ for
the $I_{4422}$ construction, implying we cannot rely on the results
beyond 3~s.f., and an alternative lifting of the same inequality was
able to achieve $61.8\%$ also. Repeating the analysis for these two
with the CVX precision variable set to ``high", gave a value of
$61.82\%$ for both liftings, which also means we cannot improve the number of
significant figures this way.

\subsubsection{$(3,5,2,2)$ Scenario}\label{sec:3522}
For the single new inequality $(3,5,2,2)$, we tested it with the CVX
precision set to high. Note that we are considering $(3,5,3,3)$
probability distributions, so lifting the inequality by adding an
extra output for each input can be done in $2^{3+5}=256$ ways. For the
first level, $\cQ^1$, we obtained an optimal threshold of $64.0\%$,
and for $\cQ^2,\cQ^3$ we obtained a threshold of $66.7\%$.  This,
together with the fact that restricting to $m_A\in\{1,3\}$ and
$m_B\in\{3,5\}$ (with reference to the form of $I_{3522}$
in~\eqref{eq:I3522}), reduces to a CHSH inequality for which the
threshold is known to be $2/3$, suggests that $\eta_{I_{3522}}=2/3$. 

The supplementary file for this case omits the inequality as only one was considered.

\subsubsection{$(3,3,3,3)$ Scenario}
For each new inequality there are $3^{3+3}=729$ possible liftings. Since it would be time 
consuming to test the detection threshold for all liftings of every inequality, we selected an ``intuitively promising" subset of liftings to test for each inequality. For each lifting we summed all coefficients corresponding to detection failure outcomes, and then tested the ten with the lowest sum for each inequality. This is to minimise the impact of the failure sub-distributions, which are entirely local. If more than ten had an equivalent sum, all were tested.\\

To cut down on computational time further, each chosen lifting was
tested with $\eta=0.65$ and discarded if no non-local distributions
$\tilde{\Pi}^{0.65}$ could be found for $\tilde{\Pi}\in\cQ^2$. 338
inequality/lifting combinations achieved a threshold below this value;
65 of them obtained value $61.8\%$ --- due to the precision of 3 s.f.\
we are unable to compare them definitively with $I_{4422}$. Testing
these 65 at level $3$ of the hierarchy we find 21 of them maintain the
value $61.8\%$. These inequalities can be found in the supplementary
files provided at~\cite{web_link}.  All of these inequalities match
the threshold $(\sqrt{5}-1)/2$ up to 3~s.f., but we have not given an
explicit quantum construction achieving this value. If it turns out
that the thresholds for $I_{4422}$ and any of these inequalities are
the same, this cannot be explained by a set of correlations common to
both, unlike in the case of $I_{3522}$ and CHSH, whose thresholds
match due to $I_{3522}$ having a CHSH submatrix (cf.\
Section~\ref{sec:3522}).

\subsection*{Added note}
These results formed part of TC's Ph.D.\ thesis~\cite{CopeThesis}
submitted to the University of York 
in September 2018.  Since then the work~\cite{OBSS2018} appeared which
independently found all the Bell inequalities in the $(4,4,2,2)$
scenario.  Ref.~\cite{CG2018} took these inequalities and performed
some analysis of them.

\subsection*{Acknowledgements}
We are grateful to Carl Miller, Yaoyun Shi, Kim Winick and Tam\'{a}s
V\'{e}rtesi for useful discussions and to Yeong-Cherng Liang for
pointers to the literature. TC and RC are supported by EPSRC's Quantum
Communications Hub (grant number EP/M013472/1). TC's Ph.D.\ was
supported by a White Rose Studentship. RC is also supported by an
EPSRC First Grant (grant number EP/P016588/1).

\bibliographystyle{naturemag}

\begin{thebibliography}{10}
\expandafter\ifx\csname url\endcsname\relax
  \def\url#1{\texttt{#1}}\fi
\expandafter\ifx\csname urlprefix\endcsname\relax\def\urlprefix{URL }\fi
\providecommand{\bibinfo}[2]{#2}
\providecommand{\eprint}[2][]{\url{#2}}

\bibitem{B1964}
\bibinfo{author}{Bell, J.}
\newblock \bibinfo{title}{On the {E}instein {P}odolsky {R}osen paradox}.
\newblock \emph{\bibinfo{journal}{Physics}} \textbf{\bibinfo{volume}{1}},
  \bibinfo{pages}{195--200} (\bibinfo{year}{1964}).

\bibitem{Aspect81}
\bibinfo{author}{Aspect, A.}, \bibinfo{author}{Grangier, P.} \&
  \bibinfo{author}{Roger, G.}
\newblock \bibinfo{title}{Experimental tests of realistic local theories via
  {B}ell's theorem}.
\newblock \emph{\bibinfo{journal}{Physical Review Letters}}
  \textbf{\bibinfo{volume}{47}}, \bibinfo{pages}{460--463}
  (\bibinfo{year}{1981}).

\bibitem{Tittel1998}
\bibinfo{author}{Tittel, W.}, \bibinfo{author}{Brendel, J.}, \bibinfo{author}{Gisin, B.}, \bibinfo{author}{Herzog, T.}, \bibinfo{author}{Zbinden, H.} \& \bibinfo{author}{Gisin, N.}
\newblock \bibinfo{title}{Experimental demonstration of quantum correlations
  over more than 10 km}.
\newblock \emph{\bibinfo{journal}{Physical Review A}}
  \textbf{\bibinfo{volume}{57}}, \bibinfo{pages}{3229--3232}
  (\bibinfo{year}{1998}).

\bibitem{Giustina&}
\bibinfo{author}{Giustina, M.}, \bibinfo{author}{Versteegh, M.~A.~M.}, \bibinfo{author}{Wengerowsky, S.}, \bibinfo{author}{Handsteiner, J.}, \bibinfo{author}{Hochrainer, A.}, \bibinfo{author}{Phelan, K.}, \bibinfo{author}{Steinlechner, F.}, \bibinfo{author}{Kofler, J.}, \bibinfo{author}{Larsson, J.~A.},  \emph{et~al.}
\newblock \bibinfo{title}{Significant-loophole-free test of {B}ell's theorem
  with entangled photons}.
\newblock \emph{\bibinfo{journal}{Physical Review Letters}}
  \textbf{\bibinfo{volume}{115}}, \bibinfo{pages}{250401}
  (\bibinfo{year}{2015}).

\bibitem{Hensen&}
\bibinfo{author}{Hensen, B.}, \bibinfo{author}{Bernien, H.}, \bibinfo{author}{Dr\'{e}au, A.~E.}, \bibinfo{author}{Reiserer, A.}, \bibinfo{author}{Kalb, N.}, \bibinfo{author}{Blok, M.~S.}, \bibinfo{author}{Ruitenberg, J.}, \bibinfo{author}{Vermeulen, R.~F.~L.}, \bibinfo{author}{Schouten, R.~N.} \emph{et~al.}
\newblock \bibinfo{title}{Loophole-free {B}ell inequality violation using
  electron spins separated by 1.3 kilometres}.
\newblock \emph{\bibinfo{journal}{Nature}} \textbf{\bibinfo{volume}{526}},
  \bibinfo{pages}{682--686} (\bibinfo{year}{2015}).

\bibitem{Shalm&}
\bibinfo{author}{Shalm, L.~K.}, \bibinfo{author}{Meyer-Scott, E.}, \bibinfo{author}{Christensen, B.~G.}, \bibinfo{author}{Bierhorst, P.}, \bibinfo{author}{Wayne, M.~A.}, \bibinfo{author}{Stevens, M.~J.}, \bibinfo{author}{Gerrits, T.}, \bibinfo{author}{Glancy, S.}, \bibinfo{author}{Hamel, D.~R.} \emph{et~al.}
\newblock \bibinfo{title}{Strong loophole-free test of local realism}.
\newblock \emph{\bibinfo{journal}{Physical Review Letters}}
  \textbf{\bibinfo{volume}{115}}, \bibinfo{pages}{250402}
  (\bibinfo{year}{2015}).

\bibitem{Ekert91}
\bibinfo{author}{Ekert, A.~K.}
\newblock \bibinfo{title}{Quantum cryptography based on {B}ell's theorem}.
\newblock \emph{\bibinfo{journal}{Physical Review Letters}}
  \textbf{\bibinfo{volume}{67}}, \bibinfo{pages}{661--663}
  (\bibinfo{year}{1991}).

\bibitem{MayersYao}
\bibinfo{author}{Mayers, D.} \& \bibinfo{author}{Yao, A.}
\newblock \bibinfo{title}{Quantum cryptography with imperfect apparatus}.
\newblock In \emph{\bibinfo{booktitle}{Proceedings of the 39th Annual Symposium
  on Foundations of Computer Science (FOCS-98)}}, \bibinfo{pages}{503--509}
  (\bibinfo{publisher}{IEEE Computer Society}, \bibinfo{address}{Los Alamitos,
  CA, USA}, \bibinfo{year}{1998}).

\bibitem{BHK}
\bibinfo{author}{Barrett, J.}, \bibinfo{author}{Hardy, L.} \&
  \bibinfo{author}{Kent, A.}
\newblock \bibinfo{title}{No signalling and quantum key distribution}.
\newblock \emph{\bibinfo{journal}{Physical Review Letters}}
  \textbf{\bibinfo{volume}{95}}, \bibinfo{pages}{010503}
  (\bibinfo{year}{2005}).

\bibitem{AGM}
\bibinfo{author}{Ac\'in, A.}, \bibinfo{author}{Gisin, N.} \&
  \bibinfo{author}{Masanes, L.}
\newblock \bibinfo{title}{From {B}ell's theorem to secure quantum key
  distribution.}
\newblock \emph{\bibinfo{journal}{Physical Review Letters}}
  \textbf{\bibinfo{volume}{97}}, \bibinfo{pages}{120405}
  (\bibinfo{year}{2006}).

\bibitem{VV2}
\bibinfo{author}{Vazirani, U.} \& \bibinfo{author}{Vidick, T.}
\newblock \bibinfo{title}{Fully device-independent quantum key distribution}.
\newblock \emph{\bibinfo{journal}{Physical Review Letters}}
  \textbf{\bibinfo{volume}{113}}, \bibinfo{pages}{140501}
  (\bibinfo{year}{2014}).

\bibitem{ColbeckThesis}
\bibinfo{author}{Colbeck, R.}
\newblock \emph{\bibinfo{title}{Quantum and Relativistic Protocols For Secure
  Multi-Party Computation}}.
\newblock Ph.D. thesis, \bibinfo{school}{University of Cambridge}
  (\bibinfo{year}{2007}).
\newblock \bibinfo{note}{Also available as \url{https://arxiv.org/abs/0911.3814}.}

\bibitem{PAMBMMOHLMM}
\bibinfo{author}{Pironio, S.} \emph{et~al.}
\newblock \bibinfo{title}{Random numbers certified by {B}ell's theorem}.
\newblock \emph{\bibinfo{journal}{Nature}} \textbf{\bibinfo{volume}{464}},
  \bibinfo{pages}{1021--1024} (\bibinfo{year}{2010}).

\bibitem{CK2}
\bibinfo{author}{Colbeck, R.} \& \bibinfo{author}{Kent, A.}
\newblock \bibinfo{title}{Private randomness expansion with untrusted devices}.
\newblock \emph{\bibinfo{journal}{Journal of Physics A}}
  \textbf{\bibinfo{volume}{44}}, \bibinfo{pages}{095305}
  (\bibinfo{year}{2011}).

\bibitem{MS1}
\bibinfo{author}{Miller, C.~A.} \& \bibinfo{author}{Shi, Y.}
\newblock \bibinfo{title}{Robust protocols for securely expanding randomness
  and distributing keys using untrusted quantum devices}.
\newblock In \emph{\bibinfo{booktitle}{Proceedings of the 46th Annual ACM
  Symposium on Theory of Computing}}, STOC '14, \bibinfo{pages}{417--426}
  (\bibinfo{publisher}{ACM}, \bibinfo{address}{New York, NY, USA},
  \bibinfo{year}{2014}).

\bibitem{BGP2010}
\bibinfo{author}{Bancal, J.-D.}, \bibinfo{author}{Gisin, N.} \&
  \bibinfo{author}{Pironio, S.}
\newblock \bibinfo{title}{Looking for symmetric {B}ell inequalities}.
\newblock \emph{\bibinfo{journal}{Journal of Physics A}}
  \textbf{\bibinfo{volume}{43}}, \bibinfo{pages}{385303}
  (\bibinfo{year}{2010}).

\bibitem{PV2009}
\bibinfo{author}{P\'{a}l, K.} \& \bibinfo{author}{V\'{e}rtesi, T.}
\newblock \bibinfo{title}{Quantum bounds on {B}ell inequalities}.
\newblock \emph{\bibinfo{journal}{Physics Review A}}
  \textbf{\bibinfo{volume}{79}}, \bibinfo{pages}{022120}
  (\bibinfo{year}{2009}).

\bibitem{DS2015}
\bibinfo{author}{Deza, M.} \& \bibinfo{author}{Sikiri\'{c}, M.}
\newblock \bibinfo{title}{Enumeration of the facets of cut polytopes over some
  highly symmetric graphs}.
\newblock \emph{\bibinfo{journal}{International Transactions in Operational
  Research}} \textbf{\bibinfo{volume}{23}}, \bibinfo{pages}{853--860}
  (\bibinfo{year}{2015}).
  
\bibitem{CHSH1969}
\bibinfo{author}{Clauser, J.}, \bibinfo{author}{Horne, M.},
  \bibinfo{author}{Shimony, A.} \& \bibinfo{author}{Holt, R.}
\newblock \bibinfo{title}{Proposed experiment to test local hidden-variable
  theories}.
\newblock \emph{\bibinfo{journal}{Physical Review Letters}}
  \textbf{\bibinfo{volume}{23}}, \bibinfo{pages}{880--884}
  (\bibinfo{year}{1969}).

\bibitem{F1981}
\bibinfo{author}{Froissart, M.}
\newblock \bibinfo{title}{Constructive generalization of {B}ell's
  inequalities}.
\newblock \emph{\bibinfo{journal}{Il Nuevo Cimento B}}
  \textbf{\bibinfo{volume}{64}}, \bibinfo{pages}{241--251}
  (\bibinfo{year}{1981}).

\bibitem{P2004}
\bibinfo{author}{Pironio, S.}
\newblock \emph{\bibinfo{title}{Aspects of Quantum Non-locality}}.
\newblock Ph.D. thesis, \bibinfo{school}{Universit\'{e} Libre de Bruxelles}
  (\bibinfo{year}{2004}).
  
\bibitem{S2003}
\bibinfo{author}{\'{S}liwa, C.}
\newblock \bibinfo{title}{Symmetries of the {B}ell correlation inequalities}.
\newblock \emph{\bibinfo{journal}{Physical Letters A}}
  \textbf{\bibinfo{volume}{317}}, \bibinfo{pages}{165-168}
  (\bibinfo{year}{2003}).
  
\bibitem{CG2004}
\bibinfo{author}{Collins, D.} \& \bibinfo{author}{Gisin, N.}
\newblock \bibinfo{title}{A relevant two qubit {B}ell inequality inequivalent
  to the {CHSH} inequality}.
\newblock \emph{\bibinfo{journal}{Journal of Physics A}}
  \textbf{\bibinfo{volume}{37}}, \bibinfo{pages}{1775--1787}
  (\bibinfo{year}{2004}).


\bibitem{CGLMP2002}
\bibinfo{author}{Collins, D.}, \bibinfo{author}{Gisin, N.},
  \bibinfo{author}{Linden, N.}, \bibinfo{author}{Massar, S.} \&
  \bibinfo{author}{Popescu, S.}
\newblock \bibinfo{title}{Bell inequalities for arbitrarily high-dimensional
  systems}.
\newblock \emph{\bibinfo{journal}{Physics Review Letters}}
  \textbf{\bibinfo{volume}{88}}, \bibinfo{pages}{040404}
  (\bibinfo{year}{2002}).

\bibitem{KKCZO2002}
\bibinfo{author}{Kaszlikowski, D.}, \bibinfo{author}{Kwek, L.~C.},
  \bibinfo{author}{Chen, J.-L.}, \bibinfo{author}{\.Zukowski, M.} \&
  \bibinfo{author}{Oh, C.~H.}
\newblock \bibinfo{title}{{C}lauser-{H}orne inequality for three-state
  systems}.
\newblock \emph{\bibinfo{journal}{Physical Review A}}
  \textbf{\bibinfo{volume}{65}}, \bibinfo{pages}{032118}
  (\bibinfo{year}{2002}).
  
\bibitem{QVB2014}
\bibinfo{author}{Quintino, M. T.}, \bibinfo{author}{V\'{e}rtesi, T.} \&
  \bibinfo{author}{Brunner, N.}
\newblock \bibinfo{title}{Joint measurability, {E}instein-{P}odolsky-{R}osen steering, and {B}ell nonlocality}.
\newblock \emph{\bibinfo{journal}{Physical Review Letters}}
  \textbf{\bibinfo{volume}{113}}, \bibinfo{pages}{160402}
  (\bibinfo{year}{2014}).

\bibitem{SBSL2016}
\bibinfo{author}{Schwarz S.} \& \bibinfo{author}{Bessire, B.} \& \bibinfo{author}{Stefanov A.} \& \bibinfo{author}{Liang, Y.-C.}
\newblock \bibinfo{title}{Bipartite {B}ell inequalities with three ternary-outcome measurements—from theory to experiments}.
\newblock \emph{\bibinfo{journal}{New Journal of Physics}} \textbf{\bibinfo{volume}{18}}, \bibinfo{pages}{035001}
  (\bibinfo{year}{2016}).


  
\bibitem{G1967}
\bibinfo{author}{Gr{\"u}nbaum, B.}
\newblock \emph{\bibinfo{title}{Convex Polytopes}}
  (\bibinfo{publisher}{Springer Science and Business Media New York Academic
  Press}, \bibinfo{year}{1967}).

\bibitem{AF1992}
\bibinfo{author}{Avis, D.} \& \bibinfo{author}{Fukuda, K.}
\newblock \bibinfo{title}{A pivoting algorithm for convex hulls and vertex
  enumeration of arrangements and polyhedra}.
\newblock \emph{\bibinfo{journal}{Discrete and Computational Geometry}}
  \textbf{\bibinfo{volume}{8}}, \bibinfo{pages}{295--313}
  (\bibinfo{year}{1992}).

\bibitem{MS1971}
\bibinfo{author}{McMullen, P.} \& \bibinfo{author}{Shephard, G.}
\newblock \emph{\bibinfo{title}{Convex Polytopes and the Upper Bound
  Conjecture}} (\bibinfo{publisher}{Cambridge University Press},
  \bibinfo{year}{1971}).

\bibitem{W1971}
\bibinfo{author}{Walsh, G.}
\newblock \emph{\bibinfo{title}{An Introduction to Linear Programming}}
  (\bibinfo{publisher}{Holt, Rinehart and Winston Ltd.}, \bibinfo{year}{1971}).

\bibitem{R1970}
\bibinfo{author}{Rockafellar, R.}
\newblock \emph{\bibinfo{title}{Convex Analysis}}
  (\bibinfo{publisher}{Princeton University Press}, \bibinfo{year}{1970}).

\bibitem{D1947}
\bibinfo{author}{Dantzig, G.}
\newblock \bibinfo{title}{Maximization of a linear function of variables
  subject to linear inequalities}.
\newblock \emph{\bibinfo{journal}{Activity Analysis of Production and
  Allocation}} \bibinfo{pages}{339--347} (\bibinfo{year}{1947}).

\bibitem{MRR2006}
\bibinfo{author}{Nocedal, J.} \& \bibinfo{author}{Wright, S.}
\newblock \emph{\bibinfo{title}{Numerical Optimization}}
  \bibinfo{publisher}{Springer Series in Operations Research and Financial
  Engineering} (\bibinfo{year}{2006}).
  
\bibitem{RBG2014}
\bibinfo{author}{Rosset, D.} \& \bibinfo{author}{Bancal, J.-D.} \& \bibinfo{author}{Gisin, N.}
\newblock \bibinfo{title}{Classifying 50 years of {B}ell inequalities}.
\newblock \emph{\bibinfo{journal}{Journal of Physics A}}
  \textbf{\bibinfo{volume}{47}}, \bibinfo{pages}{424022} (\bibinfo{year}{2014}).

\bibitem{T1993}
\bibinfo{author}{Tsirelson, B.}
\newblock \bibinfo{title}{Some results and problems on quantum {B}ell-type
  inequalities}.
\newblock \emph{\bibinfo{journal}{Hadronic Journal Supplement}}
  \textbf{\bibinfo{volume}{8}}, \bibinfo{pages}{329--345} (\bibinfo{year}{1993}).

\bibitem{P1999}
\bibinfo{author}{Peres, A.}
\newblock \bibinfo{title}{All the {B}ell Inequalities}.
\newblock \emph{\bibinfo{journal}{Foundations of Physics}}
  \textbf{\bibinfo{volume}{29}}, \bibinfo{pages}{589-614}
  (\bibinfo{year}{1999}).

\bibitem{P2005}
\bibinfo{author}{Pironio, S.}
\newblock \bibinfo{title}{Lifting {B}ell inequalities}.
\newblock \emph{\bibinfo{journal}{Journal of Mathematical Physics}}
  \textbf{\bibinfo{volume}{46}}, \bibinfo{pages}{062112}
  (\bibinfo{year}{2005}).

\bibitem{JM2005}
\bibinfo{author}{Jones, N.} \& \bibinfo{author}{Masanes, L.}
\newblock \bibinfo{title}{Interconversion of nonlocal correlations}.
\newblock \emph{\bibinfo{journal}{Physical Review A}}
  \textbf{\bibinfo{volume}{72}}, \bibinfo{pages}{052312}
  (\bibinfo{year}{2005}).

\bibitem{NPA2008}
\bibinfo{author}{Navascues, M.}, \bibinfo{author}{Pironio, S.} \&
  \bibinfo{author}{Ac\'{i}n, A.}
\newblock \bibinfo{title}{A convergent hierarchy of semidefinite programs
  characterizing the set of quantum correlation}.
\newblock \emph{\bibinfo{journal}{New Journal of Physics}}
  \textbf{\bibinfo{volume}{10}}, \bibinfo{pages}{073013}
  (\bibinfo{year}{2008}).

\bibitem{ZKBA1999}
\bibinfo{author}{\.Zukowski, M.}, \bibinfo{author}{Kaszlikowski, D.},
  \bibinfo{author}{Baturo, A.} \& \bibinfo{author}{\r{A}ke Larsson, J.}
\newblock \bibinfo{title}{Strengthening the {B}ell theorem: Conditions to
  falsify local realism in an experiment}.
\newblock
  \url{https://arxiv.org/abs/quant-ph/9910058}
   (\bibinfo{year}{1999}).

\bibitem{ABBAGP2010}
\bibinfo{author}{Almeida, M.} \emph{et~al.}
\newblock \bibinfo{title}{Guess your neighbor's input: A multipartite nonlocal
  game with no quantum advantage}.
\newblock \emph{\bibinfo{journal}{Physical Review Letters}}
  \textbf{\bibinfo{volume}{104}}, \bibinfo{pages}{230404}
  (\bibinfo{year}{2010}).

\bibitem{web_link}
\bibinfo{author}{Cope, T.} \& \bibinfo{author}{Colbeck, R.}
\newblock \bibinfo{note}{Supplementary files available at
  \url{http://www-users.york.ac.uk/~rc973/Bell}}.

\bibitem{Branciard11}
\bibinfo{author}{Branciard, C.}
\newblock \bibinfo{title}{Detection loophole in Bell experiments: How postselection modifies the requirements to observe nonlocality}.
\newblock \emph{\bibinfo{journal}{Physical Review A}}
  \textbf{\bibinfo{volume}{83}}, \bibinfo{pages}{032123}
  (\bibinfo{year}{2011}).

\bibitem{M2002}
\bibinfo{author}{Massar, S.}
\newblock \bibinfo{title}{Nonlocality, closing the detection loophole, and
  communication complexity}.
\newblock \emph{\bibinfo{journal}{Physical Review A}}
  \textbf{\bibinfo{volume}{65}}, \bibinfo{pages}{032121}
  (\bibinfo{year}{2002}).

\bibitem{C1980}
\bibinfo{author}{Cirel'son, B.}
\newblock \bibinfo{title}{Quantum generalizations of {B}ell's inequality}.
\newblock \emph{\bibinfo{journal}{Letters of Mathematical Physics}}
  \textbf{\bibinfo{volume}{4}}, \bibinfo{pages}{93} (\bibinfo{year}{1980}).

\bibitem{M1986}
\bibinfo{author}{Mermin, N.}
\newblock \emph{\bibinfo{title}{Techniques and Ideas in Quantum Measurement
  Theory}}, \bibinfo{pages}{422--428} (\bibinfo{publisher}{New York Academy of
  Sciences}, \bibinfo{year}{1986}).

\bibitem{E1993}
\bibinfo{author}{Eberhard, P.}
\newblock \bibinfo{title}{Background level and counter efficiencies required
  for a loophole-free {E}instein-{P}odolsky-{R}osen experiment}.
\newblock \emph{\bibinfo{journal}{Physical Review A}}
  \textbf{\bibinfo{volume}{47}}, \bibinfo{pages}{R747--R750}
  (\bibinfo{year}{1993}).

\bibitem{VPB2010}
\bibinfo{author}{V\'{e}rtesi, T.}, \bibinfo{author}{Pironio, S.},  \&
  \bibinfo{author}{Brunner, N.}
\newblock \bibinfo{title}{Closing the detection loophole in {B}ell experiments
  using qudits}.
\newblock \emph{\bibinfo{journal}{Physics Review Letters}}
  \textbf{\bibinfo{volume}{104}}, \bibinfo{pages}{060401}
  (\bibinfo{year}{2010}).

\bibitem{MP2003}
\bibinfo{author}{Massar, S.} \& \bibinfo{author}{Pironio, S.}
\newblock \bibinfo{title}{Violation of local realism vs detection efficiency}.
\newblock \emph{\bibinfo{journal}{Physical Review A}}
  \textbf{\bibinfo{volume}{68}}, \bibinfo{pages}{062109}
  (\bibinfo{year}{2003}).

\bibitem{cvx}
\bibinfo{author}{Grant, M.} \& \bibinfo{author}{Boyd, S.}
\newblock \bibinfo{title}{{CVX}: Matlab software for disciplined convex
  programming, version 2.1}.
\newblock \bibinfo{howpublished}{\url{http://cvxr.com/cvx}}
  (\bibinfo{year}{2014}).

\bibitem{mosek}
\bibinfo{author}{{MOSEK ApS}}.
\newblock \emph{\bibinfo{title}{The MOSEK optimization toolbox for Matlab
  manual. Version 8.1.}} (\bibinfo{year}{2017}).
\newblock \urlprefix\url{http://docs.mosek.com/8.1/toolbox/index.html}.

\bibitem{CopeThesis}
\bibinfo{author}{Cope, T.}
\newblock \emph{\bibinfo{title}{The Role of Entanglement in Quantum Communication,
 and Analysis of the Detection Loophole}}.
\newblock Ph.D. thesis, \bibinfo{school}{University of York}
  (\bibinfo{year}{2018}).
\newblock \bibinfo{note}{Also available as \url{http://etheses.whiterose.ac.uk/22955/}.}

\bibitem{OBSS2018}
\bibinfo{author}{Oudot, E.}, \bibinfo{author}{Bancal, J.-D.},
  \bibinfo{author}{Sekatski, P.},  \& \bibinfo{author}{Sangouard, N.}
\newblock \bibinfo{title}{Bipartite nonlocality with a many-body system}.
\newblock
  \url{https://arxiv.org/abs/1810.05636}
  (\bibinfo{year}{2018}).

\bibitem{CG2018}
\bibinfo{author}{Cruzeiro, E.~Z.} \& \bibinfo{author}{Gisin, N.}
\newblock \bibinfo{title}{Complete list of {B}ell inequalities with four binary
  settings}.
\newblock
  \emph{\bibinfo{journal}{\emph{\url{https://arxiv.org/abs/1811.11820}}}}
  (\bibinfo{year}{2018}).
  
\end{thebibliography}

\appendix

\section{Main mathematical results}\label{App:Proof}
\begin{lemma}\label{posformlemma}
  Let $B$ be a matrix such that $\tr(B^T\Pi)\geq c$ is a Bell
  inequality.  There is a matrix $\tilde{B}$ whose entries are all
  non-negative such that $\tr(\tilde{B}^T\Pi)\geq1$ represents the
  same Bell inequality.
\end{lemma}
\begin{proof}
  Let $\mathbf{1}$ be the matrix with every entry $1$. The matrix
  $B'=B+(1-c)\mathbf{1}/(m_Am_B)$ is such that
  $\tr((B')^T\Pi)\geq1$ is a Bell inequality.  If $B'$ has no negative
  entries we are done.  Otherwise, suppose the minimum entry of $B'$
  is $-\alpha$.  If we choose
  $\tilde{B}=(B'+\alpha\mathbf{1})/(1+\alpha m_Am_B)$, so that, by
  construction, it has no negative entries, then for any
  local $\Pi$ we have
  \begin{align*}
    \tr(\tilde{B}^T\Pi)=\frac{1}{1+\alpha
    m_Am_B}\tr((B')^T\Pi)+\frac{\alpha}{1+\alpha
    m_Am_B}\tr(\mathbf{1}\Pi)\geq\frac{1}{1+\alpha
    m_Am_B}+\frac{\alpha m_Am_B}{1+\alpha
    m_Am_B}=1\,.
  \end{align*}
\end{proof}
\begin{corollary}
  Let $B$ be a matrix such that $\tr(B^T\Pi)\leq c$ is a Bell
  inequality.  There is a matrix $\tilde{B}$ whose entries are all
  non-negative such that $\tr(\tilde{B}^T\Pi)\geq1$ represents the
  same Bell inequality.
\end{corollary}
\begin{proof}
  The original Bell inequality is equivalent to $\tr((-B)^T\Pi)\geq c$
  from which we can apply Lemma~\ref{posformlemma}.
\end{proof}
  
\begin{lemma}\label{every}
  Let $B$ be a matrix with no negative entries and such that
  $\tr(B^T\Pi)\geq 1$ is a facet Bell inequality. There exists an
  extremal no-signalling distribution $\Pi^{NS}$ such that
  $\tr(B^T\Pi^{NS})=0$.
\end{lemma}
\begin{proof}
  Since $B$ represents a violatable Bell inequality, there exists a
  no-signalling point $\tilde{\Pi}$ such that $\tr(B^T\tilde{\Pi})<1$ but
  $\tr(\hat{B}^T\tilde{\Pi})\geq 1$ for all other matrices $\hat{B}$ that
  represent facet Bell inequalities with local bound 1.  Thus, if we
  run the dual program~\eqref{eq:dual} for $\tilde{\Pi}$, the optimum is
  obtained when $M^*$ corresponds to the same Bell inequality as
  $B$ (it may be that $M^*\neq B$, but it must represent the
    same Bell inequality).
  Using the complementary slackness condition~\eqref{Slack1}, we have
  \begin{align}
    \tr\left(B^T\left(\tilde{\Pi}-\sum_i x^*_iP^{\rL,i}\right)\right)=0\,,\label{cs1}
  \end{align}
  where $\{x^*_i\}$ achieve the optimum in the primal
  problem~\eqref{eq:prim}.  This implies that the non-local part,
  $\tilde{\Pi}-\sum_i x^*_iP^{\rL,i}$, of $\tilde{\Pi}$ gives value $0$ for the Bell
  inequality $B$.  This non-local part is a convex combination of
  extremal non-local no-signalling distributions, $\{\Pi_j^{NS}\}$, each
  satisfying $\tr(B^T\Pi_j^{NS})=0$, as required.
\end{proof}

\begin{remark}\label{rmk:1}
  For the problem in the previous lemma, the second complementary
  slackness condition~\eqref{Slack2} gives that
  \begin{align}
    \sum_i\left(\tr\left(B^TP^{\rL,i}\right)-1\right)x^*_i=0.\label{cs2}
  \end{align}
  Hence, for all $i$ either $x^*_i=0$ or
  $\left(\tr\left(B^TP^{\rL,i}\right)-1\right)=0$ (both values are
  non-negative). As $x_i^*$ is non-zero if and only if the local
  distribution $P^{\rL,i}$ is in the local part of $\tilde{\Pi}$, we can
  conclude that $\tr\left(B^TP^{\rL,i}\right)=1$ for these local
  distributions, and so the local part of $\tilde{\Pi}$ satisfies the Bell
  inequality with equality.
\end{remark}
\begin{remark}\label{rmk:2}
  Remark~\ref{rmk:1} provides an insight into why it may be useful to
  remove the local part of a quantum distribution before using as the
  dual objective function. If one uses the original distribution, any
  deterministic distributions in the local part must be saturating
  points of the Bell inequality obtained by the dual. Removing the
  local part removes this requirement, which may allow the output Bell
  inequality to be such that these local points are not saturating. We
  are not aware of any bipartite Bell inequality for which there is no
  quantum violation (although these exist for three
  parties~\cite{ABBAGP2010}).  If such non-violatable inequalities exists,
  it is possible that they may be found by the dual program whose
  input is a quantum distribution with its local part removed.
\end{remark}

\begin{theorem}\label{noKLM}
  Consider the $(m_A,m_B,2,2)$ scenario and let $B$ be a matrix with
  no negative entries and such that $\tr(B^T\Pi)\geq 1$ is a facet
  Bell inequality. There exists an extremal no-signalling distribution
  $\hat{\Pi}$ that, up to relabellings, takes the form of~\eqref{eq:nsform}
  with $g=h=0$ such that $\tr(B^T\hat{\Pi})=0$.
\end{theorem}
\begin{proof}
  By Lemma~\ref{every}, there exists an extremal no-signalling
  distribution $\tilde{\Pi}$ such that $\tr(B^T\tilde{\Pi})=0$.  Suppose that $\tilde{\Pi}$
  has the form~\eqref{eq:nsform} with either $g\neq0$, $h\neq0$, or
  both.  Our aim is to show that in these cases there is always
  another extremal no-signalling distribution $\hat{\Pi}$ of the
  form~\eqref{eq:nsform} with $g=h=0$ such that $\tr(B^T\hat{\Pi})=0$.\medskip

  \noindent{\bf Case 1:} Suppose $g=0$, $h\neq 0$.  Since the
  coefficients of $B$ are non-negative, $B$ must have zero entries
  whenever $\tilde{\Pi}$ has non-zero entries.  Hence, $B(0b|xy)=0$ for
  $b\in\{0,1\}$, $x\in\{m_A-h+1,\ldots,m_A\}$ and
  $y\in\{1,\ldots,m_B\}$.
  
  Suppose there exists a local deterministic distribution, $P^1$, such
  that $\tr(B^TP^1)=1$ and for which $P^1_{A|x'}(1)=1$ for some
  $x'\in\{m_A-h+1,\ldots,m_A\}$.  Consider now another local
  deterministic distribution $P^2$ that is identical to $P^1$ except
  that $P^2_{A|x'}(0)=1$ (i.e., $P^2$ is formed by exchanging the row
  of $P^1$ corresponding to $P^1_{A|x'}(0)$ with the row corresponding
  to $P^1_{A|x'}(1)$).  It follows that
  $1\leq\tr(B^TP_2)\leq\tr(B^TP_1)=1$, i.e., $P^2$ saturates the Bell
  inequality if $P^1$ does. It also follows that the $2\times2$ blocks
  of $B$ corresponding to a measurement of $X=x'$ and any
  $Y\in\{1,\ldots,m_B\}$ have the form
  $\left(\begin{array}{cc}0&0\\0&\gamma\end{array}\right)$ or
  $\left(\begin{array}{cc}0&0\\\gamma&0\end{array}\right)$ (depending
  on whether $P^1_{B|y}(0)=1$ or $P^1_{B|y}(1)=1$), where $\gamma$ is
  an arbitrary non-negative value.

  We can therefore replace the $2\times2$ blocks of $\tilde{\Pi}$
  corresponding to $X=x'$, $Y\in\{1,\ldots,m_B\}$ by $A$ or $S$
  without affecting the value of $\tr(B^T\tilde{\Pi})$.  In other words, if
  $\tr(B^TP^1)=1$ and $P^1_{A|x'}(1)=1$ for some
  $x\in\{m_A-h+1,\ldots,m_A\}$, there exists another extremal
  no-signalling distribution, $\Pi'$, with a smaller value of $h$ that
  also has $\tr(B^T\Pi')=0$.

  If we can reduce in this way until $h=0$, we are done.
  Alternatively, we reduce to the case where $\tr(B^TP)=1$
  implies $P_{A|x}(0)=1$ for all $x\in\{m_A-h+1,\ldots,m_A\}$.  The
  affine span of the local deterministic distributions satisfying
  $\tr(B^TP)=1$ is hence the same as the affine span of those
  satisfying $\tr(\tilde{B}^T\tilde{P})=1$, where $\tilde{B}$ and
  $\tilde{P}$ comprise the first $2(m_A-h)$ rows of $B$ and $P$
  respectively.  This is at most $((m_A-h)(n_A-1)+1)(m_B(n_B-1)+1)-2$,
  and hence contradicts the assumption that $B$ is a facet
  inequality.\medskip

  \noindent{\bf Case 2:} If $g\neq0$, $h=0$, we can run an analogous
  argument.\medskip

  \noindent{\bf Case 3:} Suppose that both $g\neq0$ and $h\neq0$, and
  consider some $x'\in\{m_A-h+1,\ldots,m_A\}$ and
  $y'\in\{m_B-g+1,\ldots,m_B\}$. We must have $B(0b|x'y)=0$ for
  $b\in\{0,1\}$ and $y\in\{1,\ldots,m_B-g\}$, $B(a0|xy')=0$ for
  $a\in\{0,1\}$ and $x\in\{1,\ldots,m_A-h\}$, and $B(00|x'y')=0$.
  
  Suppose there exists a local deterministic distribution, $P^1$, such
  that $\tr(B^TP^1)=1$ and for which $P^1_{A|x'}(1)=1$. Let $P^2$ be
  another local deterministic distribution that is identical to $P^1$
  except that $P^2_{A|x'}(0)=1$ and $P^2_{B|y}(0)=1$ for all
  $y\in\{m_B-g+1,\ldots,m_B\}$. It follows that
  $1\leq\tr(B^TP_2)\leq\tr(B^TP_1)=1$, i.e., $P^2$ saturates the Bell
  inequality if $P^1$ does.

  It also follows that
  \begin{itemize}
  \item For $X=x'$ and any $Y\in\{1,\ldots,m_B-g\}$ the corresponding
    $2\times2$ blocks of $B$ have the form
    $\left(\begin{array}{cc}0&0\\0&\gamma\end{array}\right)$ or
    $\left(\begin{array}{cc}0&0\\\gamma&0\end{array}\right)$
    (depending on whether $P^1_{B|y}(0)=1$ or $P^1_{B|y}(1)=1$).
  \item For $X=x'$ and $Y\in\{m_B-g+1,\ldots,m_B\}$ the analogous blocks of $B$ have the form
    $\left(\begin{array}{cc}0&\gamma_1\\\gamma_2&0\end{array}\right)$ or
    $\left(\begin{array}{cc}0&\gamma_1\\0&\gamma_2\end{array}\right)$.
  \end{itemize}
  We can therefore replace the $2\times2$ blocks of $\tilde{\Pi}$
  corresponding to $X=x'$ by either $A$, $S$ or $L$ (whichever matches
  the zeros of $B$) without changing the value of $\tr(B^T\tilde{\Pi})$.  In
  other words, if there exists a local deterministic distribution,
  $P^1$, such that $\tr(B^TP^1)=1$ and for which $P^1_{A|x'}(1)=1$,
  then there exists an extremal no-signalling distribution, $\Pi'$,
  with a smaller value of $h$ that also satisfies $\tr(B^T\Pi')=0$.

  If we can reduce in this way until $h=0$, then either $g=0$ are we
  are done, or we can complete the argument using Case~2.
  Alternatively, we reduce to the case where $\tr(B^TP)=1$ implies
  $P_{A|x}(0)=1$ for all $x\in\{m_A-h+1,\ldots,m_A\}$, and
  $P_{B|y}(0)=1$ for all $y\in\{m_B-g+1,\ldots,m_B\}$. The affine span
  of the local deterministic distributions satisfying $\tr(B^TP)=1$ is
  hence the same as the affine span of those satisfying
  $\tr(\tilde{B}^T\tilde{P})=1$, where $\tilde{B}$ and $\tilde{P}$
  comprise the first $2(m_A-h)$ rows and $2(m_B-g)$ columns of $B$ and
  $P$ respectively.  This is at most
  $((m_A-h)(n_A-1)+1)((m_B-g)(n_B-1)+1)-2$, and hence contradicts the
  assumption that $B$ is a facet inequality.
\end{proof}

To get the idea of the proof, let us consider an example.  Take
$\tilde{\Pi}=\left(\begin{array}{ccc}S&S&L\\S&A&L\\K&K&M\end{array}\right)$.
A Bell inequality with $\tr\left(B^T\tilde{\Pi}\right)=0$ must have the form
$$B=\left(\begin{array}{cc|cc|cc}0&v_1&0&v_2&0&v_3\\ v_4&0&v_5&0&0&v_6\\
\hline
0&v_7&v_8&0&0&v_9\\ v_{10}&0&0&v_{11}&0&v_{12}\\
\hline
0&0&0&0&0&v_{13}\\ v_{14}&v_{15}&v_{16}&v_{17}&v_{18}&v_{19}\end{array}\right),$$
        where $v_i$ denotes an arbitrary non-negative entry.

        Consider now the local deterministic distributions
        $$P^1=\left(\begin{array}{cc|cc|cc}1&0&1&0&1&0\\
                      0&0&0&0&0&0\\
                      \hline
                      1&0&1&0&1&0\\
                      0&0&0&0&0&0\\
                      \hline
                      0&0&0&0&0&0\\
                      0&1&0&1&1&0\end{array}\right),\quad P^2=\left(\begin{array}{cc|cc|cc}1&0&1&0&1&0\\
                      0&0&0&0&0&0\\
                      \hline
                      1&0&1&0&1&0\\
                      0&0&0&0&0&0\\
                      \hline
                      1&0&1&0&1&0\\
                      0&0&0&0&0&0\end{array}\right),\quad \hat{P}^1=\left(\begin{array}{cc|cc|cc}1&0&1&0&0&1\\ 0&0&0&0&0&0\\
                      \hline
                      1&0&1&0&0&1\\ 0&0&0&0&0&0\\
                      \hline
                      0&0&0&0&0&0\\ 0&1&0&1&0&1\end{array}\right).$$

If $\tr(B^TP^1)=1$ then we must have $\tr(B^TP^2)=1$, and hence
$v_{15}=v_{17}=v_{18}=0$.  In this case the distribution
$\Pi'=\left(\begin{array}{ccc}S&S&L\\S&A&L\\S&S&L\end{array}\right)$
will also satisfy $\tr(B^T\Pi')=0$. We can then apply the argument of case 2 to this distribution.

Similarly, if $\tr(B^T\hat{P}^1)=1$ then we must have $\tr(B^TP^2)=1$,
and hence $v_3=v_9=v_{15}=v_{17}=v_{19}=0$.  In this case the
distribution
$\Pi'=\left(\begin{array}{ccc}S&S&A\\S&A&A\\S&S&S\end{array}\right)$
will also satisfy $\tr(B^T\Pi')=0$.  Note that by relabelling the
outputs for $y=3$ corresponds to exchanging the $A$ entries for $S$s
in the final column of $\Pi'$, bringing us into a form
matching~\eqref{eq:nsform}.

By arguments of this kind it follows that either we can reduce to a
case with a lower value of $g$ or $h$, or all the local deterministic
distributions with $\tr(B^TP)=1$ have zeros in the final row and
column ($P_{X|a=3}(0)=1$ and $P_{Y|b=3}(0)=1$).  In the latter case
the dimension of the plane containing the saturating local deterministic distributions is insufficient for $B$ to be a facet Bell inequality.

\section{Generating Quantum Distributions}\label{App:QDis}
As stated in the main body, one may use our linear programming
algorithm using quantum distributions rather than extremal
no-signalling ones. To generate these, one must first fix the
dimension of the state $d$. One then creates a normalised real vector
$\boldsymbol{\lambda}$ of $d$ non-zero real elements. We take these as
the Schmidt coefficients of the pure entangled state
$\ket{\phi}=\sum_{i=1}^d\lambda_i\ket{i}\ot\ket{i}$.  We then generate
$m_A+m_B$ random unitaries
$\left\{U^{A_1},U^{A_2},\ldots, U^{A_{m_A}},U^{B_1},U^{B_2},\ldots,
  U^{B_{m_B}}\right\}$, each unitary corresponding to a
measurement. Since the columns of each $U^i$ are orthonormal, we can
define projection operators
$P_k^i := \sum_{k\in\mathcal{S}_n^i}U^i\proj{k}(U^i)^\dagger$, where
$\{\mathcal{S}^i_n\}$ is a partition of $\{1,2,\ldots,n_A\}$ or
$\{1,2,\ldots,n_B\}$ as appropriate. These projectors satisfy
$\sum_k P^{i}_k=\left(U^i\right)^\dagger U^i=\mathbb{I}_d$. Thus, we
obtain the probability distribution:
\begin{equation}
P_{AB|xy}(a,b)=\bra{\phi}(P^{A_x}_a\otimes P^{B_y}_b)\ket{\phi}.
\end{equation}
We may then use this as the objective function. Note that this construction does not guarantee a non-local distribution.

\section{Guaranteeing All Inequality Classes}\label{App:Guar}
Although for the scenarios in which the number of inequality classes
was known, Algorithm 1 was able to generate a representative of every
class, due to the degeneracy of the optimal solutions it does not
guarantee that this generally will be the case. In this appendix we
discuss an alteration to the algorithm that can, in principle, provide
this.  However, the run time of such an algorithm is prohibitive.  We
nevertheless state the method here, because it gives the idea behind
Algorithm~1.

Suppose we have a specific Bell inequality $B$ we wish to find as a
solution to~\eqref{eq:dual}. Let $\tilde{\Pi}$ be an extremal no-signalling
distribution such that $\tr(B^T\tilde{\Pi})=0$ and let $\{P_j\}_{j=1}^d$
be a set of linearly independent local deterministic distributions
saturating $B$ (where $d$ is the dimension of the local polytope).  If
we define $\Pi'=\left(1-\delta\right)\tilde{\Pi}+\delta\sum_{j=1}^d P_j/d$,
then $\tr(B^T\Pi')=\delta$.

Furthermore, the matrix $B$ represents the unique Bell inequality that
achieves the minimum solution to the dual problem~\eqref{eq:dual} with
input $\Pi'$.  To see this, note that no matrix $M$ for which
$\tr(M^TP)\geq 1$ is a Bell inequality can give a lower value because
$\tr(M^T\Pi)$ is bounded below by 0 and $\tr(M^TP_i)\geq1$. In
addition, no other matrix $\hat{M}$ with non-negative entries for
which $\tr(\hat{M}^TP)\geq1$ is a different Bell inequality will
achieve this value because
$\tr(\hat{M}^T\Pi')\geq\delta\sum_{j=1}^d\tr(\hat{M}^TP_j)/d$.  It
cannot be that $\tr(\hat{M}^T P_j)=1$ for all $j\in\{1,\ldots,d\}$
because this would mean $\{P_j\}_{j=1}^d$ also all lie on the facet
formed by $\hat{M}$, which is only possible if the facets are
identical.

It follows that a Bell inequality of every class can be generated by
considering all possible objective functions of the form
$\left(1-\delta\right)\tilde{\Pi}+\delta\sum_{j=1}^d P_j/d$ where
$\tilde{\Pi}$ is an extremal no-signalling point of the form in
Eq.~(\ref{eq:nsform}) and $\{P_j\}_{j=1}^d$ are linearly independent
deterministic local distributions.  However, this algorithm is
impractical for all but the smallest cases.  In the case $n_A=n_B=2$,
we wish to choose $d=(m_A+1)(m_B+1)-1$ local deterministic
distributions from $2^{m_A+m_B}$ possible choices and the number of
ways to do this scales roughly as $2^{m_Am_B(m_A+m_B)}$, which is
prohibitively large. This is why in Algorithm~1 we mix each
no-signalling point with only two local deterministic distributions.
In this algorithm the local deterministic distributions that are mixed
with are taken from those that saturate the Bell inequality $M$ found
in Step~\ref{st1:3}.  Because of Remark~\ref{rmk:1}, it is possible
that the Bell inequality $M'$ that is found in Step~\ref{st1:7} is
equal to $M$, since not all degeneracies are broken.  In practice
though, it turns out this small mixing leads the linear programming
algorithm to output different inequalities. For small values of $m_A$
and $m_B$ the algorithm can be run in reasonable time.

\section{Using the Polar Dual}\label{App:Polar}
One disadvantage of our algorithm is that not all solutions of the
linear programming problem are facet inequalities. This increases the
calculation time because we have to check the affine dimension of
every solution and discard many non-facet cases. In this appendix we
consider an alternative approach which guarantees facet outputs, by
taking advantage of the \emph{polar dual} of the local polytope.
\begin{definition}
  Given a polytope $\cP\subset\mathbb{R}^{s\times t}$, its \emph{polar
    dual}\footnote{In the study of convex bodies, there are other
    types of dual, which we do not use here.} is the set of points:
\begin{equation}
\cP^\star:=\left\{B\in\mathbb{R}^{s\times t}\,\middle\vert\,\tr(B^T\Pi)\leq 1\;\forall\;\Pi\in\cP\right\}.
\end{equation}
\end{definition}
If the co-ordinate origin is interior to $\cP$ (i.e., $\Pi=0$ is
a non-boundary element of $\cP$), then the polar dual satisfies
$\left(\cP^\star\right)^\star=\cP$, and the two are
linked by the following~\cite{G1967}
\begin{align}
&\text{H-representation} & &\text{V-representation}\nonumber\\
&\cP=\left\{\Pi\in\mathbb{R}^{s\times t}\,\middle\vert\,\tr(B_j^T\Pi)\leq1\;\forall\;j\right\} & &\cP=\left\{\Pi=\sum_i\lambda_i\Pi_i\,\middle\vert\,\sum_i \lambda_i=1,\;\lambda_i\geq 0\right\}\\
&\cP^\star=\left\{B\in\mathbb{R}^{s\times t}\,\middle\vert\,\tr(B^T\Pi_i)\leq 1\;\forall\;i\right\}&&\cP^\star=\left\{B=\sum_j\lambda_jB_j\,\middle\vert\,\sum_j\lambda_j=1,\;\lambda_j\geq 0\right\}\,,
\end{align}
where $\Pi_i,B_i\in\mathbb{R}^{s\times t}$.  Hence, there is a
one-to-one correspondence between the vertices of the primal and the
facets of the dual, and vice versa.

The motivation for considering the polar dual is that by optimizing with the polar dual of the local polytope as the solution space, a simplex algorithm will always give a solution corresponding to a facet of the local polytope.

In order to use this form of the polar dual we require the origin to
lie in the interior. We hence perform a translation of coordinates. A
natural choice of the new origin is the distribution $\Pi^u$ whose
entries are all $1/n_An_B$. For example, in the $(2,2,2,2)$ scenario
this would map the extremal no-signalling distribution
$$\left(\begin{array}{cc|cc}
\frac{1}{2} & 0 & \frac{1}{2} & 0 \\
0 & \frac{1}{2}\vphantom{\frac{1}{f}} & 0 & \frac{1}{2}\\
\hline
\frac{1}{2} & 0 & 0 & \frac{1}{2} \\
0 & \frac{1}{2} & \frac{1}{2} & 0 
\end{array}\right)\ \text{to}\ \left(\begin{array}{cc|cc}
                                \phantom{-}\frac{1}{4} & -\frac{1}{4} & \phantom{-}\frac{1}{4} & -\frac{1}{4}  \\
                                -\frac{1}{4}  & \phantom{-}\frac{1}{4} & -\frac{1}{4}  & \phantom{-}\frac{1}{4}\vphantom{\frac{1}{f}}\\
 \hline
                                \phantom{-}\frac{1}{4} & -\frac{1}{4}  & -\frac{1}{4}  & \phantom{-}\frac{1}{4} \\
                                -\frac{1}{4} & \phantom{-}\frac{1}{4} & \phantom{-}\frac{1}{4} & -\frac{1}{4}
\end{array}\right).$$
Note that this representation of a distribution can have negative
entries.  The linear programming problem we wish to solve is then
\begin{align}
  \max_M\ \ &\tr(M^T\overrightarrow{\tilde{\Pi}}) \nonumber\\
  \text{subject to }\ &\tr(M^T\overrightarrow{P}^{\rL,i})\leq 1 \text{ for all }i.\label{eq:dualshift}
\end{align}
where $\overrightarrow{\tilde{\Pi}}$ refers to distribution $\tilde{\Pi}$ after shifting
origin.  The canonical form of a linear program requires positive
entries.  To take care of this, we can write $M=M_+-M_-$ where
$M_+,M_-\geq 0$ component-wise, making the problem
\begin{align}
  \max_{M_+,M_-}\ \ &\tr(M_+^T\overrightarrow{\tilde{\Pi}})-\tr(M_-^T\overrightarrow{\tilde{\Pi}}) \nonumber\\
  \text{subject to }\ &\tr(M_+^T\overrightarrow{P}^{\rL,i})-\tr(M_-^T\overrightarrow{P}^{\rL,i})\leq 1 \text{ for all }i\label{eq:dualpm}\\
  &M_+,M_-\geq 0\,.\nonumber
  \end{align}

  Although the solution space is technically unbounded due to rays of
  the form $(M_+)_{ij}=(M_-)_{ij}$, these cannot contribute to the
  objective function, and the problem~(\ref{eq:dualpm}) is
  bounded. However, our conversion of the problem into the
  form~(\ref{eq:dualshift}) means we no longer have a \emph{known}
  bound on the objective function\footnote{In the original form, we
    had that $\tr(M^T\tilde{\Pi})\geq 0$.}. To illustrate the problem with
  this, we perform the optimisation~(\ref{eq:dualpm}), for all
  extremal non-local no-signalling distributions. For the $(4,4,2,2)$
  scenario, the largest value obtained is $12/5$, and the smallest is
  $2$. (Unlike in problem~(\ref{eq:dual}), the optimal value varies
  depending on the extremal no-signalling point used.)\bigskip

We now consider a particular $(4,4,2,2)$ Bell inequality of the form:
\begin{equation}
\overrightarrow{B}_{\mathrm{ex}}=\left(
\begin{array}{cc|cc|cc|cc}
 0 & -\frac{4}{5} & 0 & 0 & 0 & 0 & 0 & -\frac{4}{5} \\
 0 & 0 & -\frac{4}{5}\vphantom{\frac{1}{f}} & 0 & 0 & 0 & 0 & 0 \\
 \hline
 0 & 0 & 0 & 0 & 0 & -\frac{4}{5} & 0 & 0 \\
 -\frac{4}{5} & 0 & 0 & -\frac{4}{5}\vphantom{\frac{1}{f}} & 0 & 0 & 0 & 0 \\
 \hline
 0 & 0 & 0 & 0 & 0 & 0 & 0 & 0 \\
 0 & 0 & 0 & 0 & 0 & 0 & 0 & 0 \\
 \hline
 0 & 0 & 0 & -\frac{4}{5} & 0 & 0 & -\frac{4}{5} & 0 \\
 0 & 0 & 0 & 0 & -\frac{4}{5} & 0 & 0 & 0 \\
\end{array}
\right),\quad \tr\left(\overrightarrow{B}_{\mathrm{ex}}^T\overrightarrow{\Pi}\right)\leq 1,
\end{equation}
and the corresponding problem of maximizing
$\tr(\overrightarrow{B}_{\mathrm{ex}}^T\overrightarrow{\tilde{\Pi}})$ over all
$\overrightarrow{\tilde{\Pi}}\in\overrightarrow{\cNS}_{(m_A,m_B,n_A,n_B)}$.
Performing this optimisation we find the maximal value to be
$9/5<2$. We can therefore conclude there is \emph{no} extremal
no-signalling point that, when used as the objective function for the
problem~(\ref{eq:dualshift}), will give the solution
$M=\overrightarrow{B}_{\mathrm{ex}}$.  By using the polar dual, if
$\eps$ is too small there are certain facets that Algorithm~1 (with
this new linear programming problem) will never output, and it is not
clear whether $\eps$ can be chosen in such a way that all facets could
in principle be output. However, this restriction does not apply to
Algorithm 2, for which this translation may be useful.

\section{Finding the Unlifted Form of a Bell Inequality}\label{App:BellDim}
Given a facet inequality $B$, we would like to know whether it is lifted
from a lower dimensional scenario. Whilst it is sometimes easy to tell
a lifting by inspection, because there are many ways to represent the
same Bell inequality (as discussed in Section~\ref{probandbell}), this
is not always the case. Here we present an algorithm that finds the
smallest scenario from which the inequality has been lifted, and gives
the inequality in the smallest scenario. This algorithm uses two
subalgorithms, the first of which checks for lifted outputs, and the
second of which removes lifted inputs.

The idea here is that output liftings correspond to copying the
coefficients corresponding to one of the existing outputs.  This means
that, in the resulting inequality, two local deterministic
distributions that are identical up to the choice of this output give
the same Bell value.

\begin{enumerate}
\item Given a facet Bell inequality and an input $x$, this checks
  whether any of the outputs are lifted. To do so the algorithm runs
  through all pairs of outputs $a_1$ and $a_2$ corresponding to the
  input $x$.  If for all local deterministic distributions that use
  the output $a_1$ for input $x$ the Bell value is the same as for the
  local deterministic distribution that is identical except that it
  uses the output $a_2$ for $x$, then the output $a_2$ can be removed
  for input $x$.  Having checked all pairs of outputs, the algorithm
  returns the set of outputs that can be removed, $O_j$.

\item Given a facet Bell inequality for which one of the inputs has
  only one possible output, this returns a new Bell inequality for the
  scenario with one fewer input without altering the bound on the Bell
  inequality and with the property that the input Bell inequality is a
  lifting of the returned inequality.  This works by adding
  no-signalling type matrices to the matrix form of the Bell
  inequality in such a way as to make all of the entries corresponding
  to the input with only one output equal to zero.  [This can be
  achieved with a no-signalling type matrix whose only non-zero
  entries are in a single row in which the entry corresponding to
  the input to be removed takes value $-v$ and the entries
  corresponding to another input with $d$ outputs has entries
  $(v,v,\ldots,v)$.] Once all the entries are zero the new Bell
  inequality is formed by removing all entries corresponding to the
  input.
\end{enumerate}

\subsubsection*{Algorithm 4}
This algorithm starts with a matrix representation $B$ of a facet Bell
inequality, and the scenario in which it was found, $S=[(n_A^1\ n_A^2\
\ldots\ n_A^{m_A})\ (n_B^1\ n_B^2\ \ldots\ n_B^{m_B})]$.
\begin{enumerate}
\item Set $j=m_A$, $S'=S$ and $B'=B$.
\item \label{start} Apply subalgorithm~1 to the
  $j^{\text{th}}$ input of Alice.
\item If the output set $O_j$ is non-empty, update $B'$ by removing the
  entries in the matrix corresponding to the outputs
  in set $O_j$ for the $j^{\text{th}}$ input and update
  $S'$ by replacing $n_A^j$ by $n_A^j-|O_j|$.
\item If the $j^{\text{th}}$ input of Alice has only one possible
  output, update $B'$ using subalgorithm~2 and update $S'$ by removing $n_A^j$.
\item Set $j=j-1$.\label{end}
\item Repeat steps \ref{start}-\ref{end} until $j=0$. 
\item Set $j=m_B$.
\item Repeat steps \ref{start}-\ref{end}, but with the roles of Alice
  and Bob interchanged, until $j=0$.
\item Return $B'$ and $S'$.
\end{enumerate}
This algorithm will return a facet Bell inequality $B'$ and a new
scenario $S'$ reduced such that $B'$ contains no lifted inputs and
outputs, and hence $S'$ is the smallest scenario for which $B$ can be
formed from a lifting. Note also that the local bounds for $B$ and $B'$
are the same and that if $B$ has all positive entries, so does $B'$.
\end{document}